\documentclass[11pt]{article}

\usepackage[margin=1in]{geometry}
\usepackage{amsmath,amssymb,amsthm,mathtools}
\usepackage{braket}
\usepackage{hyperref}
\hypersetup{hidelinks}
\usepackage{microtype}
\usepackage{authblk}
\usepackage[backend=biber]{biblatex}
\providecommand{\orgdiv}[1]{#1}
\providecommand{\orgname}[1]{#1}
\providecommand{\orgaddress}[1]{#1}
\providecommand{\street}[1]{#1}
\providecommand{\city}[1]{#1}
\providecommand{\postcode}[1]{#1}
\providecommand{\state}[1]{#1}
\providecommand{\country}[1]{#1}

\usepackage{tikz}
\usetikzlibrary{arrows.meta, positioning, shapes.geometric, calc, decorations.pathreplacing}

\newcommand{\Tr}{\operatorname{Tr}}
\newcommand{\id}{\mathbb{I}}
\newcommand{\Fid}{\mathsf{F}}        
\newcommand{\Tdist}{\mathsf{D}}      

\theoremstyle{definition}
\newtheorem{definition}{Definition}
\newtheorem{example}{Example}

\theoremstyle{plain}
\newtheorem{proposition}{Proposition}

\theoremstyle{remark}

\title{What does measuring one qubit reveal about another?\\ $K$-networks as a directed diagnostic for quantum circuits}

\author[1]{{Kostas} {Blekos}}

\author[1,2]{{Paulo Vitor} {Itabora\'i}}

\affil[1]{%
\orgdiv{Computation-Based Science and Technology Research Center},
\orgname{The Cyprus Institute},
\orgaddress{\street{20 Kavafi Street}, \city{Nicosia}, \postcode{2121}, \country{Cyprus}}}

\affil[2]{%
\orgdiv{Center for Quantum Technology and Applications (CQTA)},
\orgname{Deutsches Elektronen-Synchrotron (DESY)},
\orgaddress{\street{Platanenallee 6}, \city{Zeuthen}, \postcode{15738}, \state{Brandenburg}, \country{Germany}}}

\date{\today}

\begin{document}
\maketitle

\begin{abstract}
Many-qubit circuit states are hard to inspect directly, so they are often summarized by pairwise graph weights.
Common pairwise weights report symmetric correlations, while many circuit questions are directed and basis-specific: if qubit $i$ is measured in a given basis, how strongly does the outcome reshape the conditional state of qubit $j$?
We define $K_{i\to j}$, a directed, basis-conditioned edge weight for this question.
It is large when the two measurement outcomes occur with comparable probability and leave qubit $j$ in clearly different conditional states; it is zero when the source outcome is deterministic or the target states are indistinguishable.
The scalar uses standard binary-ensemble distinguishability; the paper's contribution is to turn this conditional comparison into a directed network layer for circuit states.
The resulting networks are computable from two-qubit reduced density matrices.
They are diagnostic (not entanglement measures): for pure two-qubit states $K$ reduces to the tangle $C^2$ (squared concurrence)~\cite{WoottersConcurrence,CKWTangle}, while separable mixed states can reach $K=1$.
Examples on teleportation, Grover, QAOA, and random circuit families show the intended use: $K$-networks map feed-forward, phase, and interaction-graph structure that symmetric or computational-basis summaries can leave weak or absent.
\end{abstract}

\section{Introduction}

When analyzing many-qubit circuits, we often want to answer questions of the form: ``If I measure qubit $i$, how much does the conditional state of qubit $j$ change?''
Standard pairwise quantities such as mutual information, concurrence, or discord quantify correlations in an undirected way.
They do not directly encode the measurement-conditioned branching that appears when the circuit has a preferred computational basis, or when readout on one qubit is used to infer information about another.

Large multi-qubit states generated by quantum algorithms are hard to interpret directly.
A common approach is to visualize \emph{pairwise} structure as a weighted graph: qubits are nodes and edges encode a chosen two-body quantity.
Recent work formalizes this idea via \emph{pairwise tomography networks} and multiplex layers (concurrence, discord, mutual information, purity, etc.)~\cite{GarciaPerezPTN}.
Mutual-information networks and complex-network diagnostics have also been used to characterize many-body structure and dynamics~\cite{ValdezMI,JonesSmallWorld}. 

This paper defines $K_{i\to j}$ as a directed edge weight for such network visualizations.
Given a measurement on qubit $i$, we compare the two conditional states produced on qubit $j$.
Their distinguishability is weighted by the measurement outcome probabilities, so a nearly deterministic measurement contributes little even if the rare branch is very different.
The resulting score is computed from two-qubit reduced density matrices (2-RDMs), making it compatible with existing pairwise tomography workflows.
The directionality of $K$ records source-target conditional dependence for the chosen measurement; it should not be read as causal structure.
This source-target asymmetry is analogous to the asymmetry between the measuring and steered parties in quantum steering~\cite{UolaSteering}.
The score is designed for visualization and basis-aligned diagnostics, complementary to correlation measures such as discord or mutual information.
Because the measurement basis is explicit, $K$ can also be used in partial tomographic settings that target a fixed readout basis or a restricted set of pairwise correlators.
In that setting, the output is a view of \textit{detectable measurement branching}, with less data than full two-qubit tomography.

This is the intent of the construction: $K$ is a basis-conditioned conditional-distinguishability network.
It scores how strongly a source-qubit measurement separates the target-qubit conditional states, giving a bounded directed layer aligned with the chosen readout basis.

\paragraph{Contributions.}
The object we propose is a directed, basis-conditioned network layer whose edges carry $K_{i\to j}$, a bounded fidelity-based score of measurement-conditioned dependence.
The scalar is close to known quantities (the complement of a classical--quantum reliability parameter, and the fidelity analogue of conditional reduced-state distance constructions).
The contribution is the layer itself: it is bounded, straightforward to estimate, and directly interpretable.
$K$ satisfies $0\le K\le 1$, has a clean characterization of when it vanishes, is invariant under target local unitaries, is covariant under source-basis rotations, and is non-increasing under channels on the target.
Writing it in the two-qubit Bloch/correlation-tensor representation gives a closed formula for each direction $K^{(\hat n)}$, a basis-optimized $K_{\max}$ (with the exact form $K_{\max}=\sigma_{\max}(T)^2$ for locally maximally mixed states), and a fixed-basis estimator that needs only a few one- and two-qubit correlators.
Two special cases connect $K$ to standard quantities: for pure two-qubit states it equals the tangle $C^2$ (squared concurrence), independent of basis, while in the balanced case it brackets the Helstrom guessing probability through the Fuchs--van de Graaf inequalities. A separable classical state can still reach $K=1$, however, so the mixed-state score is not an entanglement measure.
We compare $K$ against its closest relatives on every example in the paper (Section~\ref{subsec:comparison}), demonstrate the layer on a teleportation circuit in deferred-measurement form, where every edge has a closed form and the directed graph shows feed-forward structure that symmetric measures cannot express, and include an edge-recovery benchmark on random phase-oracle and QAOA circuits.

\subsection{Related work}
\label{subsec:related}
The closest network-level comparators are \emph{pairwise tomography networks}~\cite{GarciaPerezPTN} and \emph{mutual-information network} constructions~\cite{ValdezMI,JonesSmallWorld}.
Those works motivate pairwise network layers for visualizing many-body quantum states.
$K$ adds a layer in which the measured qubit and the target qubit play different roles, and in which the chosen measurement basis remains part of the definition.

The construction $\{(p_b,\rho_{j|b})\}_{b=0,1}$, obtained by measuring qubit $i$ and conditioning the state of $j$, is closely related to \emph{quantum steering} and to the notion of an \emph{assemblage}~\cite{UolaSteering}.
For two-qubit states, the same conditional-state geometry is often described by the steering ellipsoid: the set of Bloch vectors to which one party can steer the other by local measurements~\cite{JevticSteeringEllipsoids}.
We use the same Pauli/Bloch-matrix coordinates used in two-qubit state-space descriptions~\cite{GamelBlochMatrix}.
We emphasize that $K$ is \emph{not} a steering measure: a separable, classically correlated state reaches $K=1$ (Table~\ref{tab:sanity-checks}), so $K$ certifies neither steering nor any other form of quantumness.
The assemblage is the shared object; the scalar extracted from it here is chosen for diagnostic readability.
Among correlation quantifiers, measurement-induced nonlocality (MIN)~\cite{LuoMIN} and measurement-induced disturbance~\cite{LuoMID} are also nearby: they quantify the global disturbance caused by local measurements, and the one-way information deficit~\cite{OppenheimDeficit} the entropic cost of a local measurement, while $K_{i\to j}$ asks a narrower question: how distinguishable the two conditional states of one designated target qubit become.

At the level of the scalar itself, $K$ is a compact functional of the binary conditional ensemble $\{(p_b,\rho_{j|b})\}$, an object that is standard across several literatures, and we make no claim of a fundamentally new correlation quantity.
Three identities and near-identities show this precisely.
First, the ensemble defines a binary classical--quantum state, and the fidelity-based reliability parameter $Z(X|B)=2\sqrt{p_0p_1}\,f(\rho_{j|0},\rho_{j|1})$ of classical--quantum polar coding~\cite{RenesPolarCoding,WildeGuha}, with $f=\sqrt{\Fid}$ the root fidelity, gives the exact complement identity $K_{i\to j}=4p_0p_1-Z(X|B)^2$ (Appendix~\ref{sec:cq-reliability}).
Second, defining quantum correlations through distances between post-measurement reduced states goes back at least to Ref.~\cite{GuoAvgDist}; for a binary ensemble any squared-norm version of that construction reduces to $p_0p_1\,\|\rho_{j|0}-\rho_{j|1}\|^2$, and $K$ is its fidelity (purified-distance) analogue, normalized by the branching envelope.
Third, the entropic functional of the same ensemble is the fixed-basis one-way classical correlation (Holevo quantity)
\begin{equation}
\chi^{(\hat n)}_{i\to j} \;=\; S(\rho_j) - \sum_b p_b\, S(\rho_{j|b}),
\label{eq:chi-def}
\end{equation}
The rank-1 projective optimum of $\chi^{(\hat n)}$ is the projective-measurement version of the Henderson--Vedral classical correlation~\cite{HendersonVedral} entering quantum discord~\cite{OllivierZurekDiscord}.
The unrestricted Henderson--Vedral quantity allows the broader measurement class used in the discord literature.
Section~\ref{subsec:comparison} compares $K$ quantitatively against these quantities and against raw correlator baselines on the examples of this paper.
The paper's added object is the directed, basis-conditioned network layer built from this scalar, together with the structural results (tangle reduction, Bloch closed forms, target monotonicity, correlator estimator) that make the layer straightforward to estimate and interpret.
The claims here are narrower than a general correlation theory.
$K$ is meant to be compact, bounded, directed, and directly interpretable as conditional-state distinguishability in a chosen basis.

\subsection{Background}
Before defining $K$, we fix notation for the conditional states produced by a local measurement and for the fidelity convention used to compare them.

\subsubsection{Conditional states under a local projective measurement}\label{subsec:conditional_states}
Let $\rho$ be an $N$-qubit density operator and let $\rho_{ij} = \Tr_{\overline{ij}}[\rho]$ be the
two-qubit reduced state on qubits $(i,j)$.

Fix a rank-1 projective measurement on qubit $i$ with outcomes $b\in\{0,1\}$:
\begin{equation}
  \Pi^{(i)}_b = \ket{b}\!\bra{b}, \qquad \Pi^{(i)}_0+\Pi^{(i)}_1=\id.
\end{equation}
(We later parametrize $\{\ket{0},\ket{1}\}$ as an arbitrary Bloch-sphere basis.)

Define the unnormalized conditional state of qubit $j$ given outcome $b$ by
\begin{equation}\label{eq:unormalized_conditional_state}
  \sigma_{j|b} \;:=\; \Tr_i\!\Big[(\Pi^{(i)}_b\otimes \id_j)\,\rho_{ij}\,(\Pi^{(i)}_b\otimes \id_j)\Big],
\end{equation}
with probability $p_b := \Tr(\sigma_{j|b})$.
For $p_b>0$, the normalized conditional state is
\begin{equation}
  \rho_{j|b} := \sigma_{j|b}/p_b \qquad (p_b>0).
\end{equation}

\subsection{Fidelity and trace distance}
We use the Uhlmann--Jozsa (squared) fidelity
\begin{equation}
  \Fid(\rho,\sigma)
  :=
  \Big(\Tr\sqrt{\sqrt{\rho}\,\sigma\,\sqrt{\rho}}\Big)^2 \in [0,1],
\end{equation}
as in \cite{JozsaFidelity}.
We use the \emph{squared} Uhlmann fidelity; some references denote $\sqrt{\Fid}$ as the fidelity.
Our choice matches the convention in Eq.~\eqref{eq:fvdg} below.
A standard operationally meaningful alternative is the trace distance
\begin{equation}
  \Tdist(\rho,\sigma) := \tfrac12\|\rho-\sigma\|_1.
\end{equation}
The Fuchs--van de Graaf inequalities relate the two:
\begin{equation}
  1-\sqrt{\Fid(\rho,\sigma)} \;\le\; \Tdist(\rho,\sigma) \;\le\; \sqrt{1-\Fid(\rho,\sigma)}.
  \label{eq:fvdg}
\end{equation}

See \cite{FuchsVanDeGraaf}. These inequalities let us interpret fidelity gaps as bounds on
distinguishability in hypothesis-testing tasks (see also \cite{BaeKwekReview}).
In this paper, the relevant fidelity gap is
\begin{equation}
  \Delta_F(\rho_{j|0},\rho_{j|1})
  := 1-\Fid(\rho_{j|0}, \rho_{j|1}) \in [0,1].
  \label{eq:meas_fid_distingish}
\end{equation}
It is zero when the two conditional states are identical and equal to one when they are perfectly distinguishable.

\section{Definition of \texorpdfstring{$K$}{K} and Interpretation}

We construct $K_{i\to j}$ from the conditional states that arise when qubit $i$ is measured, weighting their distinguishability by the measurement outcome probabilities.

\begin{definition}[Conditional-distinguishability correlation score]
Fix an ordered pair of qubits $(i,j)$ in an $N$-qubit state $\rho$.
Fix a rank-1 projective measurement $\{\Pi^{(i)}_0,\Pi^{(i)}_1\}$ on qubit $i$, giving probabilities
$p_0,p_1$.
For outcomes with $p_b>0$, let $\rho_{j|b}$ be the normalized conditional state on qubit $j$.
Define
\begin{equation}
  K_{i\to j}(\rho;\{\Pi_b\})
  :=
  \begin{cases}
  4\,p_0p_1\Big(1-\Fid(\rho_{j|0},\rho_{j|1})\Big), & p_0p_1>0,\\
  0, & p_0p_1=0.
  \end{cases}
  \label{eq:Kdef}
\end{equation}
Thus a deterministic measurement outcome contributes zero by definition, and no normalized conditional state is needed for the absent branch.
\label{def:pairwiseK_Z}
\end{definition}

\begin{figure}[ht]
\centering
\includegraphics[width=.92\linewidth]{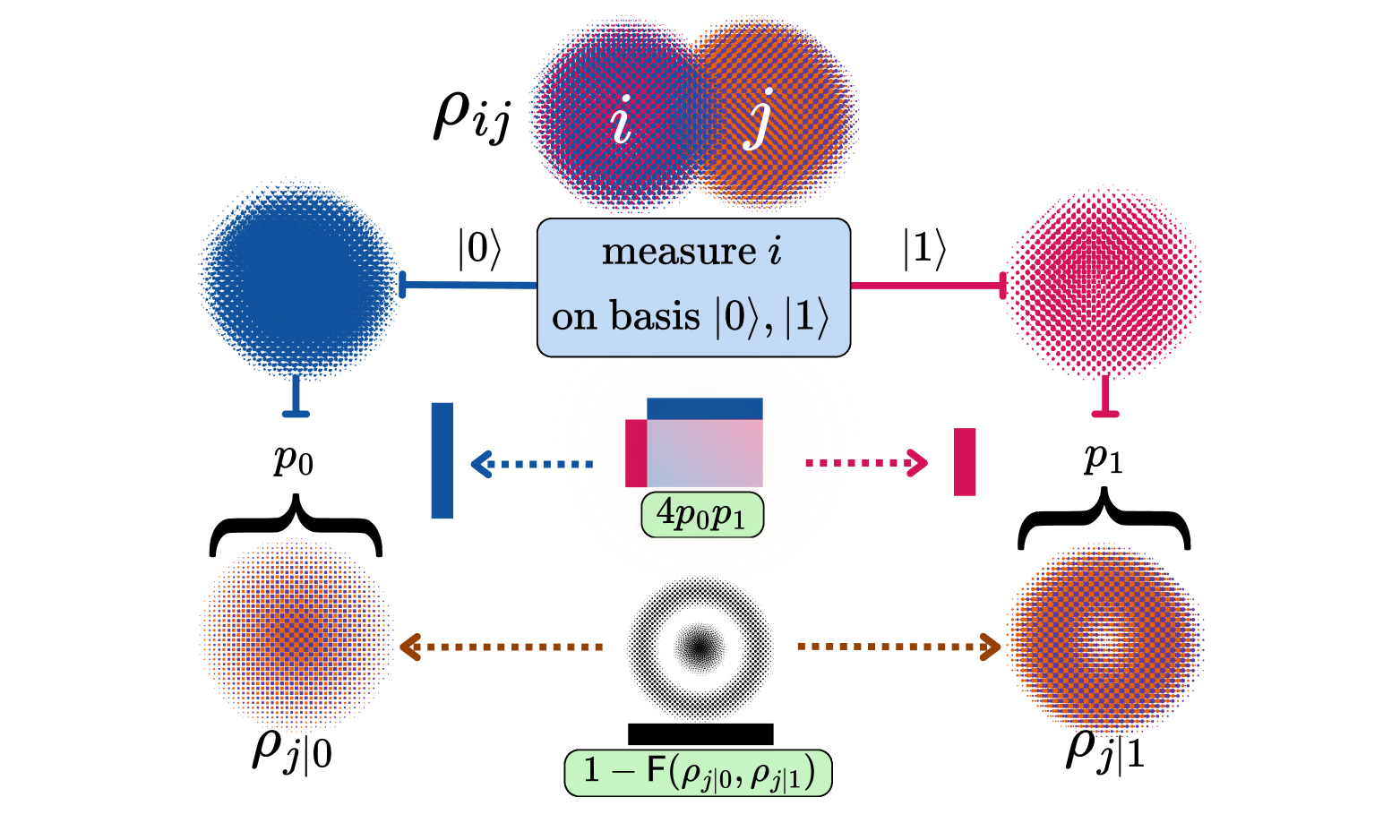}
\caption{Schematic definition of the measurement-conditioned pairwise score $K_{i\to j}$.
A two-outcome projective measurement on qubit $i$ prepares conditional states $\rho_{j|0}$ and $\rho_{j|1}$ on qubit $j$ with probabilities $p_0$ and $p_1$.
When both branches occur, the score $K$ combines the branching balance factor $4p_0p_1$ (maximal when $p_0=p_1=\tfrac12$) with the conditional-state distinguishability $1-\Fid(\rho_{j|0},\rho_{j|1})$.
If one branch is absent, Definition~\ref{def:pairwiseK_Z} sets $K=0$.}
\label{fig:K-schematic}
\end{figure}

\paragraph{Components of the score}

The factor $1-\Fid(\rho_{j|0},\rho_{j|1})$ measures conditional-state distinguishability on the target.
The factor $4p_0p_1$ records branching balance: it is maximal for $p_0=p_1=\tfrac12$ and vanishes when one branch is absent.
Thus $K=1$ exactly for balanced orthogonal branches, while $K=0$ when the two conditional states are indistinguishable or when the source measurement is deterministic.
The branching factor is a normalization \emph{choice}, made so that $K$ reports balanced branching with distinguishable targets; it is not forced by state discrimination.
In particular, $K$ should not be read as a guessing probability: a deterministic source outcome can be guessed perfectly ($P_{\mathrm{succ}}=1$) yet has $K=0$ by definition.
For imbalanced priors the pair $(p_0,p_1)$ must be retained alongside $K$; the operational discussion below makes this precise in the balanced case.

\begin{proposition}[Basic properties]
\label{prop:main-basic-properties}
For a fixed rank-1 projective measurement on the source qubit $i$, the score $K_{i\to j}$ has the following properties:
\begin{enumerate}
    \item $0\le K_{i\to j}\le 1$.
    \item If $p_0p_1>0$, then $K_{i\to j}=0$ if and only if $\rho_{j|0}=\rho_{j|1}$. If $p_0p_1=0$, then $K_{i\to j}=0$ by Definition~\ref{def:pairwiseK_Z}.
    \item If $\rho_{ij}=\rho_i\otimes\rho_j$, then $K_{i\to j}=0$.
    \item $K_{i\to j}$ is invariant under local unitaries on the target qubit $j$.
    \item $K_{i\to j}$ is non-increasing under quantum channels on the target qubit $j$.
    \item For pure two-qubit states, $K_{i\to j}=C(\ket{\psi})^2$ and is independent of the source measurement basis.
\end{enumerate}
\end{proposition}

The proofs are given in Appendices~\ref{sec:basic-properties}--\ref{sec:target-monotone}.
The pure-state reduction is special: in mixed states $K$ can be maximal on separable states, so the mixed-state score is not an entanglement measure.
Table~\ref{tab:sanity-checks} gives such a case.

\subsection{Bloch representation and computation}
\label{subsec:bloch-kmax}

The same definition can be evaluated directly from the two-qubit Bloch data.
This gives a closed expression for each measurement direction and, after optimization over directions, for $K_{\max}$.

For an ordered pair $(i,j)$, write the two-qubit reduced state in the standard Pauli expansion, with the source qubit $i$ as the first tensor factor:
\begin{equation}
\rho_{ij}
=
\frac14\left(
\id\otimes\id
+ \vec a\cdot\vec\sigma\otimes\id
+ \id\otimes \vec b\cdot\vec\sigma
+ \sum_{k,\ell\in\{x,y,z\}} T_{k\ell}\,\sigma_k\otimes\sigma_\ell
\right).
\label{eq:bloch-expansion}
\end{equation}
Here $\vec a,\vec b\in\mathbb{R}^3$ are the local Bloch vectors of $i$ and $j$, and $T\in\mathbb{R}^{3\times 3}$ is the correlation tensor:
\[
a_k=\Tr[\rho_{ij}(\sigma_k\otimes\id)],\qquad
b_\ell=\Tr[\rho_{ij}(\id\otimes\sigma_\ell)],\qquad
T_{k\ell}=\Tr[\rho_{ij}(\sigma_k\otimes\sigma_\ell)].
\]
These coordinates are also the natural coordinates for the steering-ellipsoid description of two-qubit states~\cite{JevticSteeringEllipsoids,GamelBlochMatrix}.
For $K_{j\to i}$, the same formulas apply after interchanging the two qubits, which swaps $\vec a$ and $\vec b$ and replaces $T$ by $T^{\mathsf T}$.

Let a rank-1 projective measurement on $i$ be specified by a unit vector $\hat n\in\mathbb{R}^3$:
\[
\Pi_s^{(i)}(\hat n)=\frac12(\id+s\,\hat n\cdot\vec\sigma),\qquad s\in\{+1,-1\}.
\]
For any outcome with nonzero probability,
\begin{equation}
p_s(\hat n)=\frac12(1+s\,\vec a\cdot\hat n),
\qquad
\vec r_s(\hat n)=
\frac{\vec b+s\,T^{\mathsf T}\hat n}{1+s\,\vec a\cdot\hat n},
\label{eq:conditional-bloch}
\end{equation}
where $\rho_{j|s}=\tfrac12(\id+\vec r_s\cdot\vec\sigma)$.
Thus, whenever both branches have nonzero probability,
\begin{equation}
K_{i\to j}^{(\hat n)}
=
\bigl(1-(\vec a\cdot\hat n)^2\bigr)
\left(1-\Fid\bigl(\rho(\vec r_+(\hat n)),\rho(\vec r_-(\hat n))\bigr)\right),
\label{eq:K-bloch-general}
\end{equation}
with $\rho(\vec r)=\tfrac12(\id+\vec r\cdot\vec\sigma)$.
For qubit states, the fidelity in \eqref{eq:K-bloch-general} has the closed form~\cite{JozsaFidelity}
\begin{equation}
\Fid(\rho(\vec r),\rho(\vec u))
=
\frac12\left(
1+\vec r\cdot\vec u
+\sqrt{(1-|\vec r|^2)(1-|\vec u|^2)}
\right).
\label{eq:qubit-fidelity-bloch}
\end{equation}

The basis-optimized score is
\begin{equation}
K^{\max}_{i\to j}(\rho)
:=
\max_{\|\hat n\|=1} K^{(\hat n)}_{i\to j}(\rho).
\label{eq:Kmax-def}
\end{equation}
The maximum is over rank-1 projective measurements on the source qubit $i$.

\begin{proposition}[Locally maximally mixed closed form]
\label{prop:lmm-closed-form}
Suppose the two local marginals are maximally mixed, so $\vec a=\vec b=0$ in \eqref{eq:bloch-expansion}.
Then, for every measurement direction $\hat n$,
\begin{equation}
K_{i\to j}^{(\hat n)} = |T^{\mathsf T}\hat n|^2.
\label{eq:lmm-K-direction}
\end{equation}
Consequently,
\begin{equation}
K^{\max}_{i\to j}(\rho_{ij}) = \sigma_{\max}(T)^2,
\label{eq:lmm-Kmax}
\end{equation}
where $\sigma_{\max}(T)$ is the largest singular value of $T$.
\end{proposition}

\begin{proof}
When $\vec a=\vec b=0$, the two outcomes have $p_+=p_-=\tfrac12$, and \eqref{eq:conditional-bloch} gives
\[
\vec r_+ = T^{\mathsf T}\hat n,\qquad \vec r_-=-T^{\mathsf T}\hat n.
\]
Set $\vec v=T^{\mathsf T}\hat n$.
Using \eqref{eq:qubit-fidelity-bloch},
\[
\Fid(\rho(\vec v),\rho(-\vec v))
=
\frac12\left(1-|\vec v|^2+\sqrt{(1-|\vec v|^2)^2}\right)
=1-|\vec v|^2,
\]
because $|\vec v|\le 1$ for physical conditional states.
Since $4p_+p_-=1$, \eqref{eq:lmm-K-direction} follows.
Maximizing $|T^{\mathsf T}\hat n|^2=\hat n^{\mathsf T}TT^{\mathsf T}\hat n$ over unit vectors gives the largest eigenvalue of $TT^{\mathsf T}$, which is $\sigma_{\max}(T)^2$.
\end{proof}

In the steering-ellipsoid picture, the locally maximally mixed case has an ellipsoid centered at the origin whose semi-axis lengths are the singular values of $T$.
For a fixed direction $\hat n$, the two conditional states are antipodal points on this ellipsoid.
Thus $K^{(\hat n)}$ is the squared radius reached in that direction, equivalently one quarter of the squared chord length between the two steered Bloch vectors.

\begin{proposition}[Continuity and existence of \texorpdfstring{$K_{\max}$}{Kmax}]
\label{prop:kmax-continuity}
For a fixed two-qubit state $\rho_{ij}$, define $K^{(\hat n)}_{i\to j}=0$ at directions where one outcome has probability zero.
Then $K^{(\hat n)}_{i\to j}$ is continuous on the unit sphere, and $K^{\max}_{i\to j}(\rho_{ij})$ exists.
\end{proposition}

\begin{proof}
On the open set where $p_+p_->0$, the formulas above are compositions of continuous functions.
If one branch probability tends to zero, then $4p_+p_-\to 0$, while $0\le 1-\Fid(\rho_{j|+},\rho_{j|-})\le 1$.
Therefore $K^{(\hat n)}_{i\to j}\to 0$, matching the deterministic-branch definition.
The unit sphere is compact, so the maximum exists.
\end{proof}

A practical computation of $K_{\max}$ can therefore work directly with $(\vec a,\vec b,T)$.
One evaluates \eqref{eq:K-bloch-general} on the sphere, using \eqref{eq:qubit-fidelity-bloch} for each direction, and maximizes the resulting scalar function.
For general states this is a two-dimensional optimization problem on a compact domain.
A numerical implementation should combine a deterministic spherical grid with multi-start local optimization in angular coordinates or on normalized three-vectors.
The grid gives a reproducible lower bound on $K_{\max}$, while the local optimizer refines the best candidates.
For locally maximally mixed states, Proposition~\ref{prop:lmm-closed-form} gives an exact check on the implementation.

\paragraph{Basis dependence and directionality}

The probabilities $p_0, p_1$ depend on the chosen projective measurement basis (Section~\ref{subsec:conditional_states}), so $K_{i\to j}$ is basis-dependent by design. The prefactor $4p_0p_1$ records whether that basis produces balanced branching.

Basis dependence makes $K^{(\hat n)}$ \emph{covariant}, not invariant, under unitaries on the source: conjugating the state by $U_i$ is equivalent to counter-rotating the measurement direction,
\begin{equation}
K^{(\hat n)}_{i\to j}\Big((U_i\otimes\id_j)\,\rho_{ij}\,(U_i^\dagger\otimes\id_j)\Big)
=
K^{(R_{U_i}^{-1}\hat n)}_{i\to j}(\rho_{ij}),
\label{eq:source-covariance}
\end{equation}
where $R_{U}\in SO(3)$ is the Bloch rotation implemented by $U$.
A fixed-basis layer therefore changes under source-frame changes unless the basis is co-rotated, while the basis-optimized $K^{\max}_{i\to j}$ of \eqref{eq:Kmax-def} is invariant under local unitaries on either qubit.
Note also that for locally maximally mixed states $K^{\max}_{i\to j}=K^{\max}_{j\to i}$, because $T$ and $T^{\mathsf T}$ have the same singular values (Proposition~\ref{prop:lmm-closed-form}); in that regime directional structure is a property of the fixed-basis layers, which is one reason we treat $K^{(\hat n)}$, not $K^{\max}$, as the primary object.

In general, $K_{i\to j}\neq K_{j\to i}$: measuring $i$ and examining $j$ need not yield the same score as measuring $j$ and examining $i$.
The asymmetry is determined by the two conditional ensembles induced by the chosen source-target ordering.
The resulting directed-graph structure is intentional; $K_{i\to j}$ records directed conditional dependence from $i$ to $j$, not causal influence.

As summarized in Proposition~\ref{prop:main-basic-properties} and proved in Appendix~\ref{sec:pure-twoqubit}, this asymmetry disappears for pure two-qubit states because $K$ reduces to the tangle.
This does not extend to arbitrary pure many-qubit states after tracing out the rest of the system.
Example~\ref{ex:asymmetry} gives a computational-basis counterexample.

\paragraph{What $K$ measures}

$K_{i\to j}$ is \emph{not} an entanglement measure. Entanglement measures vanish on all separable states, while $K$ can be nonzero, and indeed maximal, for certain separable mixed states where classical correlations are present (Table~\ref{tab:sanity-checks}). It quantifies \emph{measurement-conditioned dependence}: how much measuring $i$ changes the conditional state of $j$ in the chosen basis.
For the same reason, $K$ is not discord-like in the resource-theoretic sense: a fully classical state can saturate $K$ in its correlated basis while having zero discord there.
The closest comparators are the fixed-measurement quantities built from the same conditional ensemble, such as the one-way classical correlation $\chi^{(\hat n)}$ of \eqref{eq:chi-def} and conditional reduced-state distance measures~\cite{HendersonVedral,OllivierZurekDiscord,GuoAvgDist,KhalidMbQC}; see Section~\ref{subsec:related} and the quantitative comparison in Section~\ref{subsec:comparison}.
Since $K$ is built from fidelity, it inherits invariance under local unitaries on the target qubit and monotonicity under quantum channels on the target, as summarized in Proposition~\ref{prop:main-basic-properties}.
In other words, the conditional-state distinguishability that $K$ captures does not depend on the local reference frame of $j$, and noise on the target can only reduce the score.
We stress the limited domain of this monotonicity: it holds for channels on the \emph{target only}, with the source measurement fixed.
It is not an LOCC monotonicity, and preprocessing or noise on the measured source qubit changes the branching probabilities and the induced ensemble, so the score need not decrease in that direction (Appendix~\ref{sec:target-monotone}).

Relative to broader measurement-based notions such as quantum steering and measurement-induced nonlocality~\cite{UolaSteering,LuoMIN}, $K_{i\to j}$ keeps only the pairwise impact of measuring $i$ on a specified target qubit $j$ in the chosen basis.

\paragraph{Pure-state reduction}

When the two-qubit reduced state is pure, $\rho_{ij}=\ket{\psi}\!\bra{\psi}$ (that is, when qubits $i$ and $j$ are not entangled with the rest of the system or its environment), $K$ coincides with the two-qubit Coffman--Kundu--Wootters tangle $\tau_{ij}$~\cite{CKWTangle} via the concurrence squared $C(\ket{\psi})^2$~\cite{WoottersConcurrence}, and becomes independent of the measurement basis.
In this regime, $K$ reduces to a standard entanglement measure; Appendix~\ref{sec:pure-twoqubit} gives the proof.

\paragraph{Operational interpretation}

In the balanced case $p_0=p_1=\tfrac12$, the same construction has a direct discrimination interpretation.
Suppose an observer has access only to qubit $j$ and wants to guess which outcome $b\in\{0,1\}$ occurred on qubit $i$.
Appendix~\ref{sec:operational} shows that Helstrom discrimination and the Fuchs--van de Graaf inequalities give
\begin{equation}
1 - \frac{1}{2}\sqrt{1-K_{i\to j}}
\;\le\;
P_{\mathrm{succ}}^\star
\;\le\;
\frac{1}{2}\big(1+\sqrt{K_{i\to j}}\big),
\end{equation}
where $P_{\mathrm{succ}}^\star$ is the optimal probability of correctly guessing $b$ from a measurement on $j$.
Thus, for balanced branches, $K_{i\to j}=0$ corresponds to random guessing and $K_{i\to j}=1$ to perfect inference.
For imbalanced priors, $K$ alone should not be interpreted as the guessing probability; the priors $p_0,p_1$ must also be retained.

\subsection{Choice of metric: fidelity versus trace distance}
\label{subsec:metric-choice}

Definition~\ref{def:pairwiseK_Z} measures conditional-state distinguishability by the fidelity gap $1-\Fid$.
A natural variant uses the squared trace distance instead:
\begin{equation}
K^{\mathrm{tr}}_{i\to j} \;:=\; 4p_0p_1\,\Tdist(\rho_{j|0},\rho_{j|1})^2 .
\label{eq:Ktr-def}
\end{equation}
The two scores share every structural property established in this paper: boundedness, the characterization of $K=0$, invariance under target local unitaries, monotonicity under target channels (the trace distance is also contractive under CPTP maps~\cite{WatrousQIT}), the pure two-qubit reduction to the tangle (for pure conditional states $\Tdist^2=1-\Fid$), and the locally-maximally-mixed closed form of Proposition~\ref{prop:lmm-closed-form} (for antipodal Bloch vectors $\pm\vec v$ both gaps equal $|\vec v|^2$).
They are ordered, $K^{\mathrm{tr}}_{i\to j}\le K_{i\to j}$, and sandwiched by the Fuchs--van de Graaf inequalities as in \eqref{eq:K-sandwich}.
In the balanced case the trace-distance variant is exactly calibrated to discrimination, $P^\star_{\mathrm{succ}}=\tfrac12\big(1+\sqrt{K^{\mathrm{tr}}_{i\to j}}\big)$, whereas $K$ yields the two-sided bounds of \eqref{eq:Psucc-vs-K}.
Moreover, whenever the two conditional states have equal Bloch radii $|\vec r_+|=|\vec r_-|$ (equal purity), Eq.~\eqref{eq:qubit-fidelity-bloch} gives $1-\Fid=\Tdist^2$ exactly and the two scores coincide; the QAOA edges in Table~\ref{tab:comparison} are an example, while the Grover-after-oracle ensemble has branches of unequal purity and $K=2K^{\mathrm{tr}}$ there.

We adopt the fidelity version because $1-\Fid$ is the squared purified distance: this is what yields the exact complement identity with the classical--quantum reliability parameter, $K=4p_0p_1-Z(X|B)^2$ (Section~\ref{subsec:related}, Appendix~\ref{sec:cq-reliability}), connecting $K$ to a quantity with an established operational role in classical--quantum polar coding~\cite{RenesPolarCoding,WildeGuha}; no equally direct trace-distance identity is available.
The choice is otherwise inessential to the proposal: a reader who prefers exact Helstrom calibration can use $K^{\mathrm{tr}}$, obtaining $\Tdist=\tfrac12|\vec r_+-\vec r_-|$ from the same conditional Bloch vectors, with every network construction, figure, and estimator in this paper unchanged.

\subsection{Estimating \texorpdfstring{$K$}{K} from measurement data}
\label{subsec:estimating-K}

In a fixed measurement basis, the Bloch formulas specify the Pauli expectations needed to evaluate an edge.
For a Pauli measurement axis $\mu\in\{x,y,z\}$ on the source qubit $i$, let $s\in\{+1,-1\}$ label the two eigenvalue outcomes of $\sigma_\mu$.
Write
\[
m_i^{(\mu)}=\langle \sigma_\mu^{(i)}\rangle,\qquad
b_k=\langle \sigma_k^{(j)}\rangle,\qquad
c_{\mu k}^{(ij)}=\langle \sigma_\mu^{(i)}\otimes\sigma_k^{(j)}\rangle.
\]
Then the branch probabilities and target conditional Bloch components are
\begin{equation}
p_s=\frac12(1+s\,m_i^{(\mu)}),
\qquad
r_{s,k}=
\frac{b_k+s\,c_{\mu k}^{(ij)}}{1+s\,m_i^{(\mu)}},
\qquad k\in\{x,y,z\},
\label{eq:fixed-basis-estimator}
\end{equation}
whenever $p_s>0$.
The value of $K^{(\mu)}_{i\to j}$ is then obtained by inserting the two vectors $\vec r_+$ and $\vec r_-$ into the qubit fidelity formula \eqref{eq:qubit-fidelity-bloch}.
For example, a $Z$-basis edge $K^{(Z)}_{i\to j}$ requires, for that ordered pair, the source expectation $\langle Z_i\rangle$, the three target expectations $\langle X_j\rangle,\langle Y_j\rangle,\langle Z_j\rangle$, and the three correlators $\langle Z_iX_j\rangle,\langle Z_iY_j\rangle,\langle Z_iZ_j\rangle$.
This is less information than full two-qubit tomography for a fixed basis-resolved edge, although different pairs and different bases still require a measurement design that covers the corresponding Pauli products.

For many pairs or many basis choices, one can also estimate the required local Pauli expectations and two-point correlators from a shared randomized-measurement dataset.
Classical-shadow protocols are designed for this setting: a single dataset can be used to predict many observables, including local observables and two-point correlators, with sample complexity depending logarithmically on the number of target observables under the usual shadow-norm assumptions~\cite{HuangClassicalShadows}.
For $K$, this should be treated as a plug-in estimator: the measured correlators determine $\hat p_s$ and $\hat{\vec r}_s$, and the nonlinear fidelity formula is applied afterward.
Near deterministic branches, the denominators in \eqref{eq:fixed-basis-estimator} can amplify statistical error; in that regime the prefactor $4p_+p_-$ also suppresses the final score, and uncertainty should still be reported.
Because $K$ is a nonlinear function of the estimated correlators, the plug-in estimate is biased at finite shot counts, and empirical conditional Bloch vectors can fall slightly outside the Bloch ball, where \eqref{eq:qubit-fidelity-bloch} is no longer well defined.
A practical analysis should therefore (i) project any empirical vector with $|\hat{\vec r}_s|>1$ back to the unit ball before evaluating the fidelity, (ii) report bootstrap confidence intervals over measurement shots for each edge, not only point estimates, and (iii) declare edges with $\hat p_s$ below a stated threshold as $K=0$, consistent with the deterministic-branch convention, and avoid evaluating the unstable ratio in \eqref{eq:fixed-basis-estimator}.

\section{Examples}\label{sec:examples}

We illustrate the score $K$ on small states chosen to show when it vanishes or saturates, how it depends on the measurement basis, and how directionality can arise in pairwise reduced states.
Table~\ref{tab:sanity-checks} collects basic checks on the definition.
Here $\ket{\Phi^+}=(\ket{00}+\ket{11})/\sqrt2$ and
\begin{equation}
\rho_{\mathrm{cl}}=\tfrac12\Big(\ket{00}\!\bra{00}+\ket{11}\!\bra{11}\Big).
\label{eq:classicalcorr}
\end{equation}

\begin{table}[ht]
\centering
\small
\setlength{\tabcolsep}{4pt}
\renewcommand{\arraystretch}{1.2}
\begin{tabular}{p{0.24\linewidth}p{0.13\linewidth}p{0.17\linewidth}p{0.34\linewidth}}
\hline
State or pair & Basis & Value & Main point \\
\hline
$\rho_{ij}=\rho_i\otimes\rho_j$ & any & $0$ & Conditioning on $i$ leaves the target state unchanged; deterministic branches give $K=0$ by Definition~\ref{def:pairwiseK_Z}. \\
$\ket{0}_1\otimes\ket{\Phi^+}_{23}$ & $Z$ & \begin{tabular}[t]{@{}l@{}}$K_{1\to2}=0$\\$K_{2\to3}=1$\end{tabular} & The score detects the Bell pair and ignores the deterministic source qubit. \\
$\ket{W_3}$, $1\to2$ & $Z$ & $4/9$ & Imbalanced branches and partial distinguishability give an intermediate value. \\
$\ket{\mathrm{GHZ}_N}$, $i\to j$ & $Z$ & $1$ & Every pair has balanced, orthogonal conditional branches. \\
$\ket{\mathrm{GHZ}_N}$, $i\to j$ & $X$ & $0$ for $N\ge3$ & For $N\ge3$, the displayed GHZ correlations are computational-basis correlations. \\
$\rho_{\mathrm{cl}}$, $i\to j$ & $Z$ & $1$ & A separable mixed state can have maximal $K$. \\
$\rho_{\mathrm{cl}}$, $i\to j$ & $X$ & $0$ & The same state can give zero in another basis, motivating $K^{\max}_{i\to j}$. \\
\hline
\end{tabular}
\caption{Sanity checks for $K$ on small states.}
\label{tab:sanity-checks}
\end{table}

For the cases in the table, one value is enough to show the behavior.
Directionality requires comparing both orientations, so the next example computes them explicitly.

\begin{example}[Measurement Asymmetry]\label{ex:asymmetry}

Consider the global 4-qubit pure state
\begin{equation}
    \ket{\Psi}= \frac{1}{2}\Big[ \ket{0000} + \ket{0100} + \ket{1001} + \ket{1111}\Big].
\end{equation}
For $K^{(Z)}_{1\to 2}$, we obtain $p_0 = p_1 = \frac{1}{2}$, with $\rho_{2|0}=\ket{+}\!\bra{+}$ and $\rho_{2|1}= \frac{\mathbb{I}}{2}$.
Since $\Fid(\ket{+}\!\bra{+},\mathbb{I}/2)=\tfrac12$, this gives $K^{(Z)}_{1\to 2}=\frac{1}{2}$.

For $K^{(Z)}_{2\to 1}$, we obtain $p_0 = p_1 = \frac{1}{2}$, with $\rho_{1|0}=\rho_{1|1}= \frac{\mathbb{I}}{2}$.
Hence $K^{(Z)}_{2\to 1}=0$.
Thus directionality can appear in a pairwise reduced state even when the global state is pure: measuring qubit~1 changes the conditional state of qubit~2, while measuring qubit~2 leaves the conditional state of qubit~1 unchanged.

\end{example}

Figure~\ref{fig:asymmetry-K} displays the directed $K^{(Z)}$-network for the pure four-qubit state in Example~\ref{ex:asymmetry}.

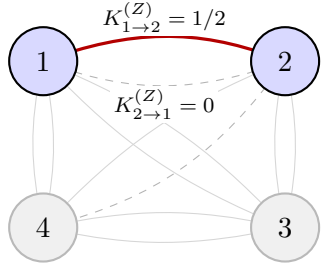
\begin{figure}[ht]
\centering
\begin{tikzpicture}[
    >=Stealth,
    qubit/.style={draw, circle, minimum size=0.9cm, thick, fill=blue!15},
    muted qubit/.style={draw=gray!55, circle, minimum size=0.9cm, thick, fill=gray!12},
    context edge/.style={->, thin, gray!32, bend left=12},
    highlighted edge/.style={->, very thick, red!70!black, bend left=20},
    zero edge/.style={->, thin, gray!60, dashed, bend left=20},
    edge label/.style={font=\scriptsize, fill=white, inner sep=1.2pt, text=black}
]

\coordinate (q1) at (0,2.2);
\coordinate (q2) at (3.2,2.2);
\coordinate (q3) at (3.2,0);
\coordinate (q4) at (0,0);

\draw[context edge] (q1) to (q3); 
\draw[context edge] (q3) to (q1); 
\draw[context edge] (q1) to (q4); 
\draw[context edge] (q4) to (q1); 
\draw[context edge] (q2) to (q3); 
\draw[context edge] (q3) to (q2); 
\draw[zero edge] (q2) to (q4); 
\draw[context edge] (q4) to (q2); 
\draw[context edge] (q3) to (q4); 
\draw[context edge] (q4) to (q3); 

\draw[highlighted edge] (q1) to node[edge label, above] {$K^{(Z)}_{1\to2}=1/2$} (q2);
\draw[zero edge] (q2) to node[edge label, below] {$K^{(Z)}_{2\to1}=0$} (q1);

\node[qubit] at (q1) {1};
\node[qubit] at (q2) {2};
\node[muted qubit] at (q3) {3};
\node[muted qubit] at (q4) {4};

\end{tikzpicture}
\caption{Directed $K^{(Z)}$-network for the four-qubit state in Example~\ref{ex:asymmetry}.
The highlighted pair has $K^{(Z)}_{1\to2}=1/2$ and $K^{(Z)}_{2\to1}=0$; the remaining arrows give the surrounding pairwise context.}
\label{fig:asymmetry-K}
\end{figure}

\section{\texorpdfstring{$K$}{K}-Networks in Circuits}
\label{sec:applications}

We apply $K$-networks to three circuit settings: Grover search and QAOA, where the pairwise scores form basis-conditioned correlation graphs through the circuit, and a teleportation circuit with deferred measurements, where the network is genuinely directed.
We close the section with a quantitative comparison between $K^{(Z)}$ and the closest existing quantities on these examples, followed by an edge-recovery benchmark on random circuit families.
Unless stated otherwise, the circuit examples use $K^{(Z)}$ and zero-based qubit labels matching the simulation convention.

\subsection{\texorpdfstring{$K$}{K}-networks: from pairwise scores to graphs}

Once each ordered pair has a score, the circuit state has a directed weighted graph representation.
Given an $N$-qubit state $\rho(t)$ at circuit depth/time $t$, one may construct a directed weighted graph $G(t)$ by taking the qubits $\{1,2,\ldots,N\}$ as nodes, assigning each directed edge $i\to j$ the weight $K^{(\hat n)}_{i\to j}(\rho(t))$ for a chosen measurement basis $\hat n$, and collecting these weights into the adjacency matrix $A_{ij}(t) = K^{(\hat n)}_{i\to j}(\rho(t))$.

This construction complements existing multiplex frameworks based on two-qubit RDM-derived quantities~\cite{GarciaPerezPTN} and mutual-information network constructions~\cite{ValdezMI,JonesSmallWorld}.
Because $K$ is explicitly measurement-conditioned and directional, it is suited to settings where a preferred measurement basis exists, such as the computational basis used by many NISQ algorithms.

By evaluating $G(t)$ at successive circuit depths, one can visualize how measurement-conditioned correlations build up or dissipate during the execution of a quantum algorithm, or while changing gate parameters.
Edge weights near 1 indicate that measuring the source qubit makes the target's conditional states nearly distinguishable, while edge weights near 0 indicate little measurement-conditioned change in the target state.

\subsection{Grover's search algorithm}

Grover's algorithm~\cite{Grover96} searches an unstructured database by repeated application of an oracle $O$ (marking the target state) and a diffuser $D$ (amplifying marked amplitudes).
Figure~\ref{fig:grover-K} shows how the $K$-network evolves through the stages of one Grover iteration for a 3-qubit search with target $\ket{101}$.
We use
\[
O_{101}=\id-2\ket{101}\!\bra{101},
\qquad
D=2\ket{+}^{\otimes 3}\!\bra{+}^{\otimes 3}-\id,
\]
with exact statevector simulation.

After the Hadamard layer, the uniform superposition has $K=0$ for all pairs: measuring any qubit provides no information about the others.
After the oracle, phase marking creates $K=0.5$ between all qubit pairs.
The oracle flips the phase of the target state, which changes the conditional states even though the computational-basis outcome probabilities are unchanged ($P(\text{target})=1/8$).
After the diffuser, these correlations are redistributed and the pairwise score drops to $K=1/8=0.125$, while the target probability increases to $P(\text{target})=25/32\approx 0.78$.

This stage-by-stage view distinguishes the role of the oracle from the role of the diffuser: the oracle creates measurement-conditioned correlations, and the diffuser converts the marked phase structure into amplitude concentration.
The $K$-network at each stage is uniform across all qubit pairs, reflecting the symmetric structure of both the oracle and diffuser.

\begin{figure}[ht]
\centering
\begin{tikzpicture}[
    >=Stealth,
    qubit/.style={draw, circle, minimum size=0.7cm, thick, fill=blue!15, font=\small},
]

\node[above, font=\footnotesize\bfseries] at (0.75, 2.0) {(a) After Hadamard};
\node[below, font=\scriptsize] at (0.75, -0.5) {$K = 0.00$};
\node[qubit] (p0q0) at (0.75,1.5) {0};
\node[qubit] (p0q1) at (0.0,0) {1};
\node[qubit] (p0q2) at (1.5,0) {2};

\node[above, font=\footnotesize\bfseries] at (3.95, 2.0) {(b) After Oracle};
\node[below, font=\scriptsize] at (3.95, -0.5) {$K = 0.50$};
\node[qubit] (p1q0) at (3.95,1.5) {0};
\node[qubit] (p1q1) at (3.2,0) {1};
\node[qubit] (p1q2) at (4.7,0) {2};

\draw[->, line width=1.35pt, red!59!black, bend left=20] (p1q0) to (p1q1);
\draw[->, line width=1.35pt, red!59!black, bend left=20] (p1q0) to (p1q2);
\draw[->, line width=1.35pt, red!59!black, bend left=20] (p1q1) to (p1q0);
\draw[->, line width=1.35pt, red!59!black, bend left=20] (p1q1) to (p1q2);
\draw[->, line width=1.35pt, red!59!black, bend left=20] (p1q2) to (p1q0);
\draw[->, line width=1.35pt, red!59!black, bend left=20] (p1q2) to (p1q1);

\node[above, font=\footnotesize\bfseries] at (7.15, 2.0) {(c) After Diffuser};
\node[below, font=\scriptsize] at (7.15, -0.5) {$K = 0.125$};
\node[qubit] (p2q0) at (7.15,1.5) {0};
\node[qubit] (p2q1) at (6.4,0) {1};
\node[qubit] (p2q2) at (7.9,0) {2};

\draw[->, line width=0.49pt, red!29!black, bend left=20] (p2q0) to (p2q1);
\draw[->, line width=0.49pt, red!29!black, bend left=20] (p2q0) to (p2q2);
\draw[->, line width=0.49pt, red!29!black, bend left=20] (p2q1) to (p2q0);
\draw[->, line width=0.49pt, red!29!black, bend left=20] (p2q1) to (p2q2);
\draw[->, line width=0.49pt, red!29!black, bend left=20] (p2q2) to (p2q0);
\draw[->, line width=0.49pt, red!29!black, bend left=20] (p2q2) to (p2q1);

\end{tikzpicture}
\caption{Evolution of the $K$-network through one Grover iteration (3 qubits, target $\ket{101}$).
(a)~After Hadamard: $K = 0$ (uniform superposition has no measurement-induced correlations).
(b)~After Oracle: $K = 0.5$ (peak correlations from phase marking).
(c)~After Diffuser: $K=1/8=0.125$ (correlations redistributed as target amplitude is amplified).}
\label{fig:grover-K}
\end{figure}
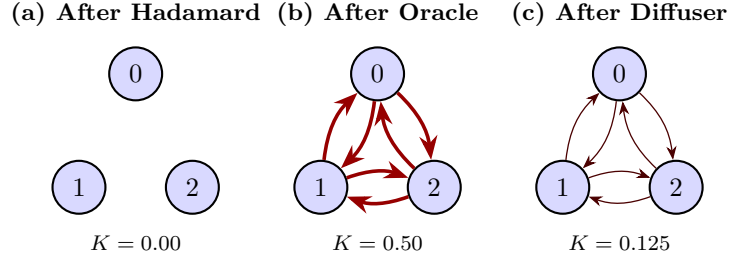

\paragraph{Comparison with classical mutual information.}
Classical mutual information $I(Z_i{:}Z_j)$ computed from the $Z$-basis outcome distribution remains zero after the oracle: the marginal and joint outcome probabilities are unchanged (each computational basis state still has probability $1/8$).
$K^{(Z)}$ jumps from $0$ to $0.5$ because the oracle creates phase structure that makes the conditional states $\rho_{j|0}$ and $\rho_{j|1}$ distinguishable, even though outcome statistics are unaffected.
This shows that $K$ can detect measurement-conditioned correlations carried by phase structure, which fixed-basis classical correlators do not see.
We emphasize the experimental point (Section~\ref{subsec:estimating-K}): this sensitivity is not available from more samples of the same computational-basis distribution; estimating $K^{(Z)}$ requires the cross-basis correlators $\langle Z_iX_j\rangle,\langle Z_iY_j\rangle$, and here the oracle's imprint appears as $|\langle Z_iX_j\rangle|=1/2$ (Table~\ref{tab:comparison}).

\subsection{QAOA on a ring graph}

Consider the Quantum Approximate Optimization Algorithm (QAOA)~\cite{FarhiQAOA} applied to MaxCut on a 4-qubit ring graph with edges $\{(0,1),(1,2),(2,3),(3,0)\}$.
A single QAOA layer ($p=1$) is applied to the initial state $\ket{+}^{\otimes 4}$ as
\[
U(\gamma,\beta)
=
\exp\!\left(-i\beta\sum_i X_i\right)
\exp\!\left(-i\gamma\sum_{\langle i,j\rangle} Z_iZ_j\right).
\]
In the circuit convention used for the simulation, this is implemented by applying $R_{ZZ}(2\gamma)$ on each ring edge, followed by $R_X(2\beta)$ on each qubit.

Figure~\ref{fig:qaoa-K} shows the $K$-network for $\gamma=0.4$ and $\beta=0.2$, computed by exact statevector simulation.
For these parameters, ring-neighbor pairs have larger $K$ values ($K=0.2498$, shown as $0.25$), while the two diagonal pairs have smaller values ($K=0.0378$, shown as $0.04$).
The directed-edge notation reflects the operational asymmetry of measurement: $K_{i\to j}$ uses $i$ as the measured source and $j$ as the target.
For this symmetric circuit, the same operations are applied across the qubits (the preparation gates and the nearest-neighbor $R_{ZZ}$ and $R_X$ layers), and the computed scores satisfy $K_{i\to j}=K_{j\to i}$.
The basis dependence still distinguishes $K$ from symmetric measures like mutual information, as discussed below; Section~\ref{subsec:teleportation} shows a circuit where the pairwise symmetry itself is genuinely broken.

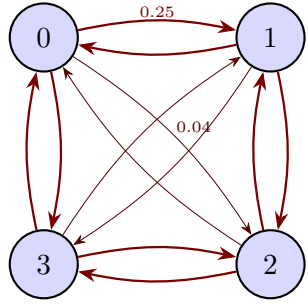
\begin{figure}[ht!]
\centering
\begin{tikzpicture}[
    >=Stealth,
    qubit/.style={draw, circle, minimum size=0.9cm, thick, fill=blue!15},
    ringedge/.style={->, line width=0.85pt, red!47!black, bend left=12},
    diagonaledge/.style={->, line width=0.38pt, red!32!black, bend left=12},
]

\node[qubit] (q0) at (0,3.0) {0};
\node[qubit] (q1) at (3.0,3.0) {1};
\node[qubit] (q2) at (3.0,0) {2};
\node[qubit] (q3) at (0,0) {3};

\draw[ringedge] (q0) to node[auto, font=\tiny, inner sep=1pt] {0.25} (q1);
\draw[diagonaledge] (q0) to node[auto, font=\tiny, inner sep=1pt] {0.04} (q2);
\draw[ringedge] (q0) to (q3);
\draw[ringedge] (q1) to (q0);
\draw[ringedge] (q1) to (q2);
\draw[diagonaledge] (q1) to (q3);
\draw[diagonaledge] (q2) to (q0);
\draw[ringedge] (q2) to (q1);
\draw[ringedge] (q2) to (q3);
\draw[ringedge] (q3) to (q0);
\draw[diagonaledge] (q3) to (q1);
\draw[ringedge] (q3) to (q2);
\end{tikzpicture}
\caption{$K$-network for a 4-qubit QAOA circuit ($p=1$, $\gamma=0.4$, $\beta=0.2$) on a ring MaxCut problem.
Edge thickness and color indicate $K_{i\to j}$ values; arrows show directionality, and the two displayed labels give representative values.
Ring neighbors (0--1, 1--2, 2--3, 3--0) have $K \approx 0.25$, while diagonal pairs (0--2, 1--3) have $K \approx 0.04$.}
\label{fig:qaoa-K}
\end{figure}

\paragraph{Comparison with mutual information.}
Figure~\ref{fig:qaoa-K-vs-MI} compares the $K$-network with the mutual-information (MI) network for the same QAOA state.
The MI network uses undirected edges weighted by $I(i{:}j) = S(\rho_i) + S(\rho_j) - S(\rho_{ij})$, a standard measure of total correlations.

The $K$-network distinguishes ring neighbors ($K=0.2498$) from diagonal pairs ($K=0.0378$), while the MI network is nearly uniform: $I=0.417$ bits on ring edges and $I=0.404$ bits on diagonal pairs.
This difference arises because $K$ is \emph{basis-conditioned}: it captures correlations visible in the computational basis, which aligns with the problem structure ($ZZ$ interactions on ring edges).
MI, being basis-independent, aggregates all correlations regardless of how they manifest under measurement.

The two layers play complementary roles: MI provides a basis-agnostic summary of total correlations, while $K$ shows which correlations appear under computational-basis measurements.
For algorithm analysis where the readout basis is fixed, $K$ can provide more direct structural information.

\begin{figure}[ht]
\centering
\begin{tikzpicture}[
    >=Stealth,
    qubit/.style={draw, circle, minimum size=0.8cm, thick, fill=blue!15, font=\small},
]

\begin{scope}[local bounding box=knet]
\node[above, font=\small\bfseries] at (0.75, 2.1) {(a) $K$-network};
\node[qubit] (k0) at (0,1.5) {0};
\node[qubit] (k1) at (1.5,1.5) {1};
\node[qubit] (k2) at (1.5,0) {2};
\node[qubit] (k3) at (0,0) {3};

\draw[->, line width=0.85pt, red!47!black, bend left=12] (k0) to (k1);
\draw[->, line width=0.28pt, red!32!black, bend left=12] (k0) to (k2);
\draw[->, line width=0.85pt, red!47!black, bend left=12] (k0) to (k3);
\draw[->, line width=0.85pt, red!47!black, bend left=12] (k1) to (k0);
\draw[->, line width=0.85pt, red!47!black, bend left=12] (k1) to (k2);
\draw[->, line width=0.28pt, red!32!black, bend left=12] (k1) to (k3);
\draw[->, line width=0.28pt, red!32!black, bend left=12] (k2) to (k0);
\draw[->, line width=0.85pt, red!47!black, bend left=12] (k2) to (k1);
\draw[->, line width=0.85pt, red!47!black, bend left=12] (k2) to (k3);
\draw[->, line width=0.85pt, red!47!black, bend left=12] (k3) to (k0);
\draw[->, line width=0.28pt, red!32!black, bend left=12] (k3) to (k1);
\draw[->, line width=0.85pt, red!47!black, bend left=12] (k3) to (k2);
\end{scope}

\begin{scope}[xshift=4cm, local bounding box=minet]
\node[above, font=\small\bfseries] at (0.75, 2.1) {(b) MI-network};
\node[qubit] (m0) at (0,1.5) {0};
\node[qubit] (m1) at (1.5,1.5) {1};
\node[qubit] (m2) at (1.5,0) {2};
\node[qubit] (m3) at (0,0) {3};

\draw[-, line width=2.50pt, blue!100!black] (m0) -- (m1);
\draw[-, line width=2.46pt, blue!98!black] (m0) -- (m2);
\draw[-, line width=2.50pt, blue!99!black] (m0) -- (m3);
\draw[-, line width=2.50pt, blue!99!black] (m1) -- (m2);
\draw[-, line width=2.46pt, blue!98!black] (m1) -- (m3);
\draw[-, line width=2.50pt, blue!99!black] (m2) -- (m3);
\end{scope}

\node[below, font=\scriptsize, align=center] at (0.75, -0.6) {Directed edges\\$K_{i\to j}$};
\node[below, font=\scriptsize, align=center] at (4.75, -0.6) {Undirected edges\\$I(i{:}j)$};
\end{tikzpicture}
\caption{Comparison of $K$-network and MI-network for the same 4-qubit QAOA state.
(a) The $K$-network (directed, red edges) distinguishes ring neighbors from diagonal pairs.
(b) The MI-network (undirected, blue edges) shows nearly uniform weights across all pairs.
Edge thickness is proportional to $K$ or $I$; MI is normalized to its maximum value for visualization.}
\label{fig:qaoa-K-vs-MI}
\end{figure}
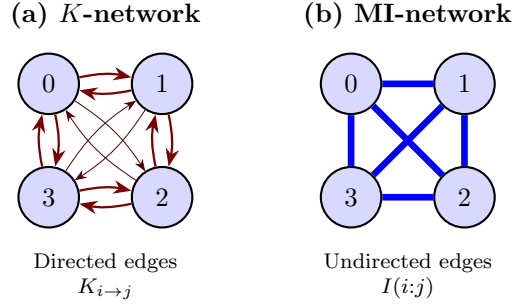

\subsection{Directed edges from deferred measurements: teleportation}
\label{subsec:teleportation}

The Grover and QAOA circuits above are symmetric, and the computed scores satisfy $K_{i\to j}=K_{j\to i}$.
In those examples, arrows record the source-target convention used to compute each conditional ensemble.
Mid-circuit measurements are the setting where direction becomes operational: the measured qubit is a genuine source, and the conditional states of the remaining qubits are exactly what subsequent feed-forward acts on.
We illustrate this with the standard teleportation circuit, written in deferred-measurement form.

Qubit $0$ carries the input state $\ket{\psi}$; qubits $1$ and $2$ hold a Bell pair (Hadamard on $1$, CNOT $1\to2$); the Bell-basis rotation applies CNOT $0\to1$ followed by a Hadamard on $0$.
In the protocol, qubits $0$ and $1$ are then measured in the $Z$ basis, and the outcomes are fed forward to qubit $2$: outcome $b_1$ of qubit $1$ triggers the correction $X^{b_1}$ and outcome $b_0$ of qubit $0$ the correction $Z^{b_0}$.
The pre-measurement state is
\begin{equation}
\ket{\Phi}
=
\frac12 \sum_{b_0,b_1\in\{0,1\}} \ket{b_0}\ket{b_1}\otimes X^{b_1}Z^{b_0}\ket{\psi}.
\label{eq:teleport-state}
\end{equation}

All six directed $Z$-basis edges of $\ket{\Phi}$ have closed forms.
Every single-qubit measurement is balanced ($p_0=p_1=\tfrac12$; the marginal of qubit $2$ is maximally mixed because the four correction branches twirl $\ket{\psi}$).
Reading the conditional ensembles off \eqref{eq:teleport-state}, measuring qubit $0$ leaves qubit $1$ with the conditional Bloch vectors $(\pm\braket{X}_\psi,0,0)$ and qubit $2$ with the branch-averaged vectors $(\pm\braket{X}_\psi,0,0)$; measuring qubit $1$ (or $2$) leaves the other with $(0,0,\pm\braket{Z}_\psi)$; and the two conditional states of qubit $0$ are identical for either choice of source.
With the equal-radius case of \eqref{eq:qubit-fidelity-bloch} (where $1-\Fid$ equals the squared half-chord, Section~\ref{subsec:metric-choice}),
\begin{equation}
K^{(Z)}_{0\to1}=K^{(Z)}_{0\to2}=\braket{X}_\psi^2,
\qquad
K^{(Z)}_{1\to2}=K^{(Z)}_{2\to1}=\braket{Z}_\psi^2,
\qquad
K^{(Z)}_{1\to0}=K^{(Z)}_{2\to0}=0 .
\label{eq:teleport-K}
\end{equation}
Figure~\ref{fig:teleport-K} shows the network for the input $\ket{\psi}=R_Y(\pi/3)\ket{0}$, where $\braket{X}_\psi^2=3/4$ and $\braket{Z}_\psi^2=1/4$.

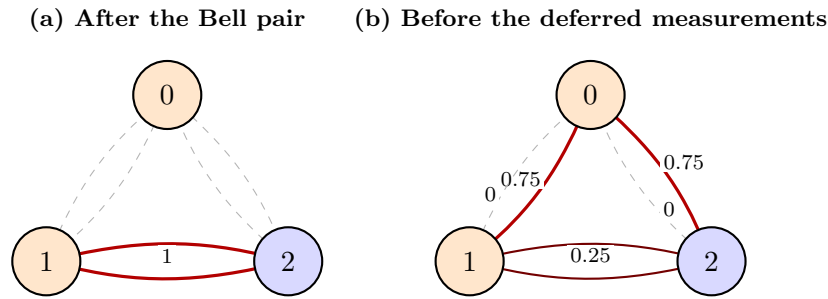
\begin{figure}[ht]
\centering
\begin{tikzpicture}[
    >=Stealth,
    qubit/.style={draw, circle, minimum size=0.9cm, thick, fill=blue!15},
    meas qubit/.style={draw, circle, minimum size=0.9cm, thick, fill=orange!20},
    strong edge/.style={->, very thick, red!70!black, bend left=14},
    mid edge/.style={->, thick, red!45!black, bend left=14},
    zero edge/.style={->, thin, gray!60, dashed, bend left=14},
    edge label/.style={font=\scriptsize, fill=white, inner sep=1.2pt, text=black}
]

\begin{scope}[local bounding box=panelA]
\node[above, font=\footnotesize\bfseries] at (1.6, 2.9) {(a) After the Bell pair};
\coordinate (a1) at (1.6, 2.2);
\coordinate (a2) at (0, 0);
\coordinate (a3) at (3.2, 0);

\draw[strong edge] (a2) to node[edge label, below] {$1$} (a3);
\draw[strong edge] (a3) to (a2);
\draw[zero edge] (a1) to (a2);
\draw[zero edge] (a2) to (a1);
\draw[zero edge] (a1) to (a3);
\draw[zero edge] (a3) to (a1);

\node[meas qubit] at (a1) {0};
\node[meas qubit] at (a2) {1};
\node[qubit] at (a3) {2};
\end{scope}

\begin{scope}[xshift=5.6cm, local bounding box=panelB]
\node[above, font=\footnotesize\bfseries] at (1.6, 2.9) {(b) Before the deferred measurements};
\coordinate (b1) at (1.6, 2.2);
\coordinate (b2) at (0, 0);
\coordinate (b3) at (3.2, 0);

\draw[strong edge] (b1) to node[edge label, above left=-2pt] {$0.75$} (b2);
\draw[strong edge] (b1) to node[edge label, above right=-2pt] {$0.75$} (b3);
\draw[mid edge] (b2) to node[edge label, below] {$0.25$} (b3);
\draw[mid edge] (b3) to (b2);
\draw[zero edge] (b2) to node[edge label, pos=0.35, left=1pt] {$0$} (b1);
\draw[zero edge] (b3) to node[edge label, pos=0.35, right=1pt] {$0$} (b1);

\node[meas qubit] at (b1) {0};
\node[meas qubit] at (b2) {1};
\node[qubit] at (b3) {2};
\end{scope}

\end{tikzpicture}
\caption{$K^{(Z)}$-network for the teleportation circuit in deferred-measurement form, input $\ket{\psi}=R_Y(\pi/3)\ket{0}$ on qubit 0.
Orange nodes ($0$, $1$) are the qubits measured in the protocol; node $2$ is the output.
(a)~After the Bell-pair preparation, the only structure is the symmetric maximal edge of the Bell pair.
(b)~Before the deferred measurements, the network is genuinely directed: $K^{(Z)}_{0\to1}=K^{(Z)}_{0\to2}=3/4$, $K^{(Z)}_{1\to2}=K^{(Z)}_{2\to1}=1/4$, and no edge enters qubit $0$ (dashed).}
\label{fig:teleport-K}
\end{figure}

Figure~\ref{fig:teleport-K}(b) has an operational interpretation.
The edge $K^{(Z)}_{0\to2}$ (respectively $K^{(Z)}_{1\to2}$) quantifies how strongly the pending $Z$ (respectively $X$) correction reshapes the output's conditional state, and for a real-amplitude input the two pending corrections split the input's Bloch information exactly: $\sin^2\varphi+\cos^2\varphi=1$.
The empty column into qubit $0$ states that no single-qubit readout anywhere in the register reveals anything about qubit $0$'s pending outcome.
Symmetric layers display none of this: quantum mutual information is nearly uniform here ($I(0{:}1)=I(0{:}2)\approx0.81$, $I(1{:}2)\approx1.19$ bits for the input shown), and the classical $Z$-basis mutual information vanishes identically on the $(0,1)$ and $(0,2)$ pairs even though $K^{(Z)}_{0\to1}=K^{(Z)}_{0\to2}=3/4$.
As in the Grover example, the directed structure is carried by conditional coherences, visible in correlators such as $\langle Z_0X_j\rangle=\braket{X}_\psi$.
For a general pure input the same expressions hold with $\braket{X}_\psi^2+\braket{Z}_\psi^2=1-\braket{Y}_\psi^2$; the $\braket{Y}_\psi^2$ deficit is invisible to every single-qubit target in this basis, an instance of the multi-qubit-target phenomenon discussed in Section~\ref{subsec:variants}.

\subsection{Comparison with the closest existing quantities}
\label{subsec:comparison}

The examples above allow a direct quantitative comparison between $K^{(Z)}$ and the quantities a reader might reach for first.
For an ordered pair $(i,j)$ and a $Z$-basis measurement on the source, Table~\ref{tab:comparison} reports the score $K^{(Z)}$ together with the trace-distance variant $K^{\mathrm{tr}}$ of \eqref{eq:Ktr-def}, the optimal probability $P^{\star}_{\mathrm{succ}}$ of guessing the source outcome from the target (Eq.~\eqref{eq:helstrom}), the fixed-basis one-way classical correlation $\chi^{(Z)}_{i\to j}$ of \eqref{eq:chi-def}, the quantum mutual information $I(i{:}j)$, the classical mutual information $I_Z$ of the joint $Z$-outcome distribution, and the Pauli correlators $|\langle Z_iZ_j\rangle|$ and $|\langle Z_iX_j\rangle|$.
The rows are the classically correlated state $\rho_{\mathrm{cl}}$ of \eqref{eq:classicalcorr}, the three-qubit Grover state after the phase oracle, and a ring and a diagonal edge of the QAOA state of Figure~\ref{fig:qaoa-K}.
All four rows have balanced branching ($p_0=p_1=\tfrac12$), so the values are directly comparable.

\begin{table}[ht]
\centering
\small
\setlength{\tabcolsep}{4.5pt}
\renewcommand{\arraystretch}{1.25}
\begin{tabular}{lcccccccc}
\hline
State / edge & $K^{(Z)}$ & $K^{\mathrm{tr}}$ & $P^{\star}_{\mathrm{succ}}$ & $\chi^{(Z)}$ & $I(i{:}j)$ & $I_Z$ & $|\langle Z_iZ_j\rangle|$ & $|\langle Z_iX_j\rangle|$ \\
\hline
$\rho_{\mathrm{cl}}$, any direction        & $1$     & $1$     & $1$     & $1$     & $1$     & $1$     & $1$     & $0$ \\
Grover after oracle, any pair              & $0.5$   & $0.25$  & $0.75$  & $0.311$ & $0.811$ & $0$     & $0$     & $0.5$ \\
QAOA ring edge $(0,1)$                     & $0.250$ & $0.250$ & $0.750$ & $0.209$ & $0.417$ & $0.095$ & $0.359$ & $0$ \\
QAOA diagonal edge $(0,2)$                 & $0.038$ & $0.038$ & $0.597$ & $0.030$ & $0.404$ & $0.004$ & $0.076$ & $0$ \\
\hline
\end{tabular}
\caption{$K^{(Z)}$ against the closest existing quantities on the examples of this paper ($Z$-basis source measurement; entropic quantities in bits).
Exact statevector values; entries are rounded.}
\label{tab:comparison}
\end{table}

Three observations summarize the table.
First, on the metric choice (Section~\ref{subsec:metric-choice}): the QAOA rows have $K^{(Z)}=K^{\mathrm{tr}}$ exactly, because the two conditional branches have equal purity, while the Grover ensemble has branches of unequal purity and $K^{(Z)}=2K^{\mathrm{tr}}$; in every row the balanced-case calibration $P^{\star}_{\mathrm{succ}}=\tfrac12(1+\sqrt{K^{\mathrm{tr}}})$ holds exactly, and $K^{(Z)}$ brackets it through \eqref{eq:Psucc-vs-K}.
Second, the entropic functional of the same ensemble, $\chi^{(Z)}$, induces the same ordering as $K^{(Z)}$ on these states; $K$ is its geometric, closed-form counterpart, and both are maximal on the classical state $\rho_{\mathrm{cl}}$, so neither certifies quantumness.
Third, no single standard column reproduces the pattern of $K^{(Z)}$ across rows.
$|\langle Z_iZ_j\rangle|$ and $I_Z$ vanish on the Grover row where $K^{(Z)}=0.5$; $|\langle Z_iX_j\rangle|$ vanishes on the QAOA rows where $K^{(Z)}$ separates ring from diagonal edges by a factor of $6.6$; and the quantum mutual information is nearly uniform across the two QAOA edges ($0.417$ versus $0.404$ bits) where $K^{(Z)}$ resolves the problem graph.
Each $K^{(Z)}$ edge is assembled from exactly the correlators of Section~\ref{subsec:estimating-K}, so the input data is ordinary. Its aggregation into a single bounded, basis-conditioned conditional-state weight reflects the circuit structure across all four cases.

\subsection{Edge-recovery benchmark}
\label{subsec:edge-recovery-benchmark}

The examples above isolate individual circuit mechanisms, so we also test whether the same edge weight recovers the generating interaction graph across random instances.
For each Erd\H{o}s--R\'enyi graph $G$ on $n$ qubits, graph edges are positives and non-edges are negatives.
In this setting, each pair score is evidence for or against an interaction edge in the graph used to generate the circuit.
A fixed cutoff would add an extra convention, since the scale of the scores can change with the circuit family, noise level, and intended visualization.
We therefore keep the scores continuous and evaluate whether true interaction edges rank above non-edges.
For each unordered pair we compute five scores: the symmetrized $K^{(Z)}$ score, quantum mutual information $I(i{:}j)$, the classical $Z$-basis mutual information $I_Z$, the absolute computational-basis correlator $|\langle Z_iZ_j\rangle|$, and a transverse correlator baseline chosen to give the raw-correlator comparison its best case
\begin{equation}
C_{Z\perp}(i,j)
=
\frac12\left[
\sqrt{\langle Z_iX_j\rangle^2+\langle Z_iY_j\rangle^2}
+
\sqrt{\langle Z_jX_i\rangle^2+\langle Z_jY_i\rangle^2}
\right].
\label{eq:zperp-baseline}
\end{equation}
The transverse magnitude is included to avoid depending on an arbitrary $X/Y$ phase convention; it is the natural low-cost baseline for phase information visible through the correlators used by \eqref{eq:fixed-basis-estimator}.
AUC gives a cutoff-free summary of this ordering.
Equivalently, AUC is the Mann--Whitney probability that a random true edge receives a larger score than a random non-edge, with ties counted as one half.
All entries below use exact statevector simulation over 200 independently sampled Erd\H{o}s--R\'enyi graphs for each row; empty and complete graphs are resampled so that both positive and negative pairs are present.
The AUC computation uses the tie convention above and no numerical optimization.

Table~\ref{tab:edge-recovery-benchmark} shows two families.
The first is a diagonal Ising phase oracle,
\[
\ket{+}^{\otimes n}\longmapsto
\prod_{(i,j)\in E(G)} \exp(-i\gamma Z_iZ_j/2)\ket{+}^{\otimes n},
\]
with $\gamma=0.4$.
Because the circuit is diagonal, the computational-basis distribution remains exactly uniform, so $I_Z$ and $|\langle Z_iZ_j\rangle|$ are at chance by construction.
Both $K^{(Z)}$ and $C_{Z\perp}$ recover the graph with AUC $1$ at this angle.
In this phase-oracle family, the same signal is also visible in the transverse correlator baseline.
The second family is a full $p=1$ QAOA layer, using the $R_{ZZ}(2\gamma)$ then $R_X(2\beta)$ convention of Section~\ref{sec:applications}, with $\gamma=1.2$ and $\beta=0.3$.
Here the mixer makes the $Z$-basis baselines nontrivial.
In these rows, $K^{(Z)}$ has higher AUC than mutual information and the $Z$-only baselines.

\begin{table}[ht]
\centering
\scriptsize
\setlength{\tabcolsep}{3.5pt}
\renewcommand{\arraystretch}{1.18}
\begin{tabular}{llccccc}
\hline
Family & $(n,p_{\mathrm{edge}})$ & $K^{(Z)}$ & $I(i{:}j)$ & $I_Z$ & $|\langle Z_iZ_j\rangle|$ & $C_{Z\perp}$ \\
\hline
Phase oracle & $(4,0.4)$ & $1.000\pm0.000$ & $1.000\pm0.000$ & $0.500\pm0.000$ & $0.500\pm0.000$ & $1.000\pm0.000$ \\
Phase oracle & $(4,0.6)$ & $1.000\pm0.000$ & $1.000\pm0.000$ & $0.500\pm0.000$ & $0.500\pm0.000$ & $1.000\pm0.000$ \\
Phase oracle & $(6,0.4)$ & $1.000\pm0.000$ & $0.996\pm0.024$ & $0.500\pm0.000$ & $0.500\pm0.000$ & $1.000\pm0.000$ \\
Phase oracle & $(6,0.6)$ & $1.000\pm0.000$ & $0.977\pm0.093$ & $0.500\pm0.000$ & $0.500\pm0.000$ & $1.000\pm0.000$ \\
Phase oracle & $(8,0.4)$ & $1.000\pm0.000$ & $0.995\pm0.016$ & $0.500\pm0.000$ & $0.500\pm0.000$ & $1.000\pm0.000$ \\
Phase oracle & $(8,0.6)$ & $1.000\pm0.000$ & $0.968\pm0.064$ & $0.500\pm0.000$ & $0.500\pm0.000$ & $1.000\pm0.000$ \\
\hline
QAOA & $(4,0.5)$ & $1.000\pm0.000$ & $1.000\pm0.000$ & $0.859\pm0.181$ & $0.859\pm0.181$ & $0.841\pm0.311$ \\
QAOA & $(6,0.5)$ & $0.980\pm0.061$ & $0.954\pm0.101$ & $0.792\pm0.161$ & $0.792\pm0.161$ & $0.841\pm0.165$ \\
QAOA & $(8,0.5)$ & $0.976\pm0.050$ & $0.934\pm0.095$ & $0.807\pm0.122$ & $0.807\pm0.122$ & $0.851\pm0.130$ \\
\hline
\end{tabular}
\caption{Edge-recovery AUC on random graph instances (mean $\pm$ standard deviation over 200 instances).
Phase-oracle rows use $\gamma=0.4$; QAOA rows use $\gamma=1.2,\beta=0.3$.
For the phase-oracle family, $I_Z$ and $|\langle Z_iZ_j\rangle|$ are exactly blind because the $Z$-basis distribution is uniform.
The transverse baseline $C_{Z\perp}$ is defined in \eqref{eq:zperp-baseline}.}
\label{tab:edge-recovery-benchmark}
\end{table}

The phase angle matters.
At $\gamma=\pi/2$ in the same $n=6$, $p_{\mathrm{edge}}=0.5$ phase-oracle family, the AUC for both $K^{(Z)}$ and $C_{Z\perp}$ drops to $0.570$, while $I_Z$ and $|\langle Z_iZ_j\rangle|$ remain exactly at $0.500$.
Figure~\ref{fig:phase-oracle-auc} shows the dependence on $\gamma$.
The point is that $K$ is a readout-basis-conditioned network weight tied to conditional-state distinguishability.

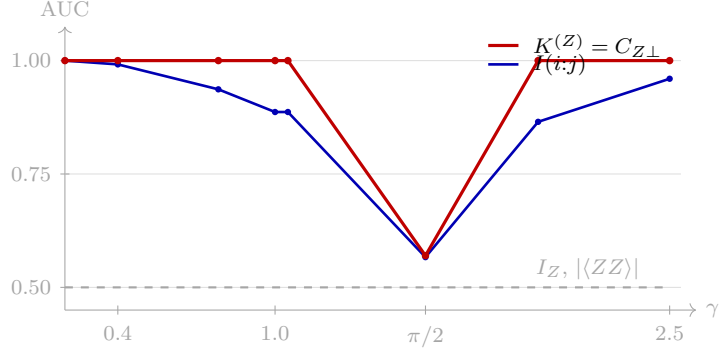
\begin{figure}[ht]
\centering
\begin{tikzpicture}[x=1cm,y=1cm,font=\scriptsize]
\draw[->,gray!70] (0,0) -- (8.35,0) node[right] {$\gamma$};
\draw[->,gray!70] (0,0) -- (0,3.75) node[above] {AUC};

\foreach \y/\lab in {0.30/0.50,1.80/0.75,3.30/1.00} {
  \draw[gray!25] (0,\y) -- (8.15,\y);
  \draw[gray!70] (-0.06,\y) -- (0.06,\y) node[left=3pt] {\lab};
}

\draw[gray!70] (0.70,-0.06) -- (0.70,0.06) node[below=4pt] {0.4};
\draw[gray!70] (2.78,-0.06) -- (2.78,0.06) node[below=4pt] {1.0};
\draw[gray!70] (4.77,-0.06) -- (4.77,0.06) node[below=4pt] {$\pi/2$};
\draw[gray!70] (8.00,-0.06) -- (8.00,0.06) node[below=4pt] {2.5};

\draw[gray!70,dashed,line width=0.8pt] (0,0.30) -- (8.00,0.30);
\node[gray!70,anchor=west] at (6.10,0.55) {$I_Z$, $|\langle ZZ\rangle|$};

\draw[blue!70!black,line width=1.0pt]
  (0.00,3.30) -- (0.70,3.25) -- (2.03,2.92) -- (2.78,2.62)
  -- (2.95,2.62) -- (4.77,0.70) -- (6.26,2.49) -- (8.00,3.06);
\foreach \x/\y in {0.00/3.30,0.70/3.25,2.03/2.92,2.78/2.62,2.95/2.62,4.77/0.70,6.26/2.49,8.00/3.06}
  \fill[blue!70!black] (\x,\y) circle (1.2pt);

\draw[red!75!black,line width=1.2pt]
  (0.00,3.30) -- (0.70,3.30) -- (2.03,3.30) -- (2.78,3.30)
  -- (2.95,3.30) -- (4.77,0.72) -- (6.26,3.30) -- (8.00,3.30);
\foreach \x/\y in {0.00/3.30,0.70/3.30,2.03/3.30,2.78/3.30,2.95/3.30,4.77/0.72,6.26/3.30,8.00/3.30}
  \fill[red!75!black] (\x,\y) circle (1.4pt);

\draw[red!75!black,line width=1.2pt] (5.60,3.50) -- (6.00,3.50);
\node[anchor=west] at (6.06,3.50) {$K^{(Z)}=C_{Z\perp}$};
\draw[blue!70!black,line width=1.0pt] (5.60,3.25) -- (6.00,3.25);
\node[anchor=west] at (6.06,3.25) {$I(i{:}j)$};
\end{tikzpicture}
\caption{Phase-oracle edge-recovery AUC versus phase angle for $n=6$, $p_{\mathrm{edge}}=0.5$ over 200 random graph instances.
The $Z$-basis baselines are exactly chance for all angles because the computational-basis distribution remains uniform.
At these angles $K^{(Z)}$ and the transverse correlator baseline $C_{Z\perp}$ coincide, while both show an aliasing dip near $\gamma=\pi/2$.}
\label{fig:phase-oracle-auc}
\end{figure}

\section{Discussion and Outlook}

The preceding results make $K$ a bounded, directed score for basis-conditioned branching in pairwise reduced states.
The score is invariant under local unitaries on the target, decreases under channels on the target, and can be computed from Bloch/correlation-tensor data.
In fixed Pauli bases, it can be estimated from a small set of one- and two-qubit correlators.
For pure two-qubit states, it equals squared concurrence.
In mixed states, it measures conditional-state distinguishability and can be maximal on separable states.

\subsection{Scope and limitations}
\label{subsec:limitations}

The directed notation should be interpreted narrowly.
$K_{i\to j}$ is directional because the source qubit is measured and the target qubit is conditioned; it is not a causal arrow and it need not follow the control-target orientation of gates in a circuit.
A controlled unitary can be symmetric in the induced conditional-state scores, while a deferred-measurement or feed-forward construction can produce a genuinely directed $K$-network.

The score is also pairwise and target-local.
It can miss information that is present only in a joint subsystem: a parity syndrome, for example, may reveal nothing about either data qubit separately while revealing the two-qubit parity perfectly.
The multi-target extension in Section~\ref{subsec:variants} addresses this in principle, at the cost of a larger target-state fidelity estimation problem.

Finally, $K$ aggregates ordinary Pauli data; the information in an edge is exactly what those correlators contain.
When the relevant transverse correlator is known in advance, a baseline such as $C_{Z\perp}$ matches $K$ on phase-oracle edge recovery.
What $K$ adds there is the bounded conditional-state interpretation and the single uniform network layer.

\subsection{Variants and extensions}
\label{subsec:variants}
The directed edge $K_{i\to j}^{(\hat n)}$ is the primary object.
Several variants follow from the same definition.
For basis-resolved analyses, one can keep separate Pauli layers $K^{(X)}$, $K^{(Y)}$, and $K^{(Z)}$.
A basis-optimized layer uses $K^{\max}_{i\to j}$ as the edge weight.
When an undirected summary is useful, one can average the two directions:
\begin{equation}
K^{\mathrm{sym}}_{ij}(\rho(t);\hat n)
= \tfrac12\big(K^{(\hat n)}_{i\to j} + K^{(\hat n)}_{j\to i}\big).
\end{equation}

The target can also be enlarged from one qubit $j$ to a subsystem $T$.
After measuring $i$, the conditional target states are $\rho_{T|0}$ and $\rho_{T|1}$, and the analogous score is
\[
4p_0p_1\bigl(1-\Fid(\rho_{T|0},\rho_{T|1})\bigr)
\]
when both branches occur, with value $0$ for deterministic branching.
This extension changes the experimental and visualization cost because fidelity estimation on a multi-qubit subsystem requires more information than the single-qubit Bloch-vector formula.
A mid-circuit parity check shows why the extension matters: after a coherent $R_X(\theta)$ error on one qubit of a Bell pair, a syndrome ancilla $a$ holding the $ZZ$ parity has $K^{(Z)}_{a\to j}=0$ for \emph{each individual} data qubit $j$, while the two-qubit-target score $K^{(Z)}_{a\to T}$ equals $\sin^2\theta$: parity information is invisible to every single-qubit target but maximal on the joint target.
We leave systematic multi-target networks to future work.

\subsection{Considerations for a practical implementation}
Equations~\eqref{eq:conditional-bloch}--\eqref{eq:qubit-fidelity-bloch} reduce the computation of a basis-resolved edge to elementary operations on Bloch vectors and a $3\times 3$ correlation tensor.
This is useful both in simulation, where $(\vec a,\vec b,T)$ can be extracted directly from a two-qubit RDM, and in experiments, where the required quantities are local Pauli expectations and two-point Pauli correlators.
For $K_{\max}$, Proposition~\ref{prop:kmax-continuity} ensures that a maximum exists, while Proposition~\ref{prop:lmm-closed-form} gives an exact benchmark for locally maximally mixed states.
For general states, the numerical optimization should report the grid resolution, number of starts, and agreement with any available independent grid check.
The induced binary cq ensemble also relates $K$ to a reliability parameter used in quantum polar coding; Section~\ref{subsec:related} discusses the relation and Appendix~\ref{sec:cq-reliability} records the identity.

\subsection{Outlook}
$K_{i\to j}$ is a measurement-based conditional-distinguishability score for analyzing correlation structure in quantum circuits.
The main use case is a family of $K$-networks, possibly basis-resolved, that can be compared with mutual-information and tomography-network layers.
This comparison is especially relevant for circuits with mid-circuit measurements, post-selection, ancilla-mediated structure, controlled gates, steering-like behavior, or changing quantum reference frames; the teleportation network of Section~\ref{subsec:teleportation} is a first example of this regime.

Table~\ref{tab:comparison} and the edge-recovery benchmark of Table~\ref{tab:edge-recovery-benchmark} give a first quantitative comparison.
They remain small exact-state studies.
The next step is a broader systematic study across graph families, depths, noise models, and shot budgets.
Such a study should ask where $K$ adds information beyond existing pairwise metrics, where it agrees with structure already visible to raw correlators, and how much experimental data is needed to estimate the nonlinear edge weights reliably.
The most important extensions are finite-shot error bars for the correlator estimator, hardware data with a fixed readout basis, and multi-target versions for parity and syndrome information that is invisible to every single-qubit target.

\section*{Acknowledgements}

This work is funded by the European Union’s Horizon Europe Framework Programme (HORIZON) under the ERA Chair scheme with grant agreement no. 101087126.

P.I. is supported with funds from the Ministry of Science, Research and Culture of the State of Brandenburg within the Centre for Quantum Technologies and Applications (CQTA). 
\begin{center}
    \includegraphics[width = 0.08\textwidth]{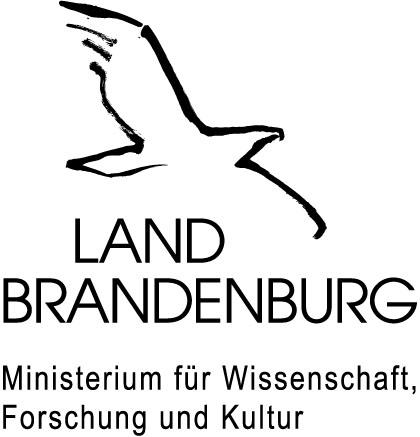}
\end{center}

\section*{Author Contributions}

K.B.: Conceptualization, Methodology, Formal analysis, Investigation, Visualization, Writing -- original draft, Writing -- review \& editing.
P.V.I.: Formal analysis, Investigation, Visualization, Writing -- original draft, Writing -- review \& editing.


\printbibliography

\section*{Appendix}
\appendix
\section{Basic properties}
\label{sec:basic-properties}

\begin{proposition}[Bounds]\label{prop:bounds}
For any $\rho$, $(i,j)$ and any rank-1 projective measurement on $i$,
\begin{equation}
  0 \le K_{i\to j}(\rho;\{\Pi_b\}) \le 1.
\end{equation}
\end{proposition}
\begin{proof}
If $p_0p_1=0$, then $K_{i\to j}=0$ by Definition~\ref{def:pairwiseK_Z}.
Assume $p_0p_1>0$.
We have $0\le \Fid(\rho_{j|0},\rho_{j|1})\le 1$ and $0\le p_0p_1\le \tfrac14$,
so $0\le 4p_0p_1(1-\Fid)\le 4(\tfrac14)\cdot 1 = 1$.
\end{proof}

\begin{proposition}[Characterization of $K=0$]\label{prop:zero-characterization}
Assume $p_0,p_1>0$. Then
\[
K_{i\to j}(\rho;\{\Pi_b\})=0
\quad \Longleftrightarrow \quad
\rho_{j|0}=\rho_{j|1}.
\]
If $p_0=0$ or $p_1=0$, then $K=0$ by Definition~\ref{def:pairwiseK_Z}.
\end{proposition}
\begin{proof}
With $p_0p_1>0$, $K=0$ iff $\Fid(\rho_{j|0},\rho_{j|1})=1$.
For density operators, $\Fid(\rho,\sigma)=1$ iff $\rho=\sigma$ \cite{JozsaFidelity}.
\end{proof}

\begin{proposition}[Local-unitary invariance on the target qubit]\label{prop:local-unitary}
Let $U$ be a single-qubit unitary on qubit $j$ and define
$\rho' := (\id_i\otimes U_j)\,\rho_{ij}\,(\id_i\otimes U_j^\dagger)$.
Then
\[
K_{i\to j}(\rho';\{\Pi_b\}) = K_{i\to j}(\rho;\{\Pi_b\}).
\]
\end{proposition}
\begin{proof}
The measurement on $i$ produces the same probabilities $p_b$.
The conditional states transform as $\rho'_{j|b}=U\rho_{j|b}U^\dagger$.
Fidelity is invariant under unitary conjugation \cite{JozsaFidelity}, hence $K$ is unchanged.
\end{proof}

\begin{proposition}[Vanishing on product states]\label{prop:product-zero}
If $\rho_{ij}=\rho_i\otimes\rho_j$ is a product state, then for \emph{any} projective measurement on $i$,
\[
K_{i\to j}(\rho;\{\Pi_b\})=0.
\]
\end{proposition}
\begin{proof}
For $\rho_{ij}=\rho_i\otimes\rho_j$, conditioning on an outcome on $i$ does not affect the state of $j$:
$\rho_{j|0}=\rho_{j|1}=\rho_j$ whenever both branches occur. If one branch is absent, Definition~\ref{def:pairwiseK_Z} sets $K=0$.
Hence $K=0$ in all cases.
\end{proof}

\section{Pure two-qubit states: \texorpdfstring{$K$}{K} is basis-independent and equals the tangle}
\label{sec:pure-twoqubit}

\begin{proposition}[Pure two-qubit reduction]
\label{prop:pure-reduction}
Let $\rho_{ij} = \ket{\psi}\!\bra{\psi}$ be a pure state on two qubits $(i,j)$.
Fix \emph{any} rank-1 projective measurement $\{\Pi^{(i)}_0,\Pi^{(i)}_1\}$ on qubit $i$,
and define $p_b,\rho_{j|b}$ and $K_{i\to j}(\rho;\{\Pi_b\})$ as in Def.~1.
Then
\begin{equation}
K_{i\to j}\!\left(\ket{\psi}\!\bra{\psi};\{\Pi_b\}\right) \;=\; 4\,\det(\rho_i),
\label{eq:K-4det}
\end{equation}
where $\rho_i = \Tr_j(\ket{\psi}\!\bra{\psi})$ is the single-qubit reduced state.
In particular, for pure two-qubit states $K_{i\to j}$ is independent of the chosen measurement basis.

Moreover, the right-hand side equals the (two-qubit) \emph{tangle}:
\begin{equation}
K_{i\to j}\!\left(\ket{\psi}\!\bra{\psi};\{\Pi_b\}\right) \;=\; \tau_{ij}
\;=\; C(\ket{\psi})^2,
\end{equation}
where $\tau_{ij}$ denotes the Coffman--Kundu--Wootters tangle~\cite{CKWTangle} and $C(\ket{\psi})$ the
two-qubit concurrence~\cite{WoottersConcurrence}.
\end{proposition}

\begin{proof}
Let $\{\ket{0},\ket{1}\}$ denote the measurement basis on qubit $i$,
so that $\Pi^{(i)}_b = \ket{b}\!\bra{b}$.
Define the unnormalized conditional kets on qubit $j$ by
\[
\ket{\psi_b} \;:=\; \big(\bra{b}\otimes I\big)\ket{\psi}\in \mathbb{C}^2.
\]
Then $p_b=\braket{\psi_b|\psi_b}$ and, for $p_b>0$,
\[
\rho_{j|b} = \ket{\phi_b}\!\bra{\phi_b},\qquad \ket{\phi_b}=\ket{\psi_b}/\sqrt{p_b}.
\]
Because the conditional states are pure,
\[
\Fid(\rho_{j|0},\rho_{j|1}) = |\braket{\phi_0|\phi_1}|^2
= \frac{|\braket{\psi_0|\psi_1}|^2}{p_0 p_1}
\quad (p_0 p_1>0).
\]
Now observe that
\[
p_b = \bra{b}\rho_i\ket{b},\qquad
\braket{\psi_0|\psi_1} = \bra{0}\rho_i\ket{1},
\]
where $\rho_i=\Tr_j(\ket{\psi}\!\bra{\psi})$.
Writing the $2\times 2$ matrix of $\rho_i$ in this basis as
\[
\rho_i =
\begin{pmatrix}
p_0 & c\\
c^* & p_1
\end{pmatrix}
\quad\text{with } c=\bra{0}\rho_i\ket{1},
\]
we have $\det(\rho_i)=p_0p_1-|c|^2$.
Plugging into Def.~1 yields (for $p_0p_1>0$)
\[
K_{i\to j} = 4p_0p_1\!\left(1-\frac{|c|^2}{p_0p_1}\right)=4(p_0p_1-|c|^2)=4\det(\rho_i).
\]
If $p_0p_1=0$, then $K_{i\to j}=0$ by Definition~\ref{def:pairwiseK_Z}, and also $\det(\rho_i)=0$,
so \eqref{eq:K-4det} holds in all cases.

Finally, for a pure two-qubit state the CKW tangle satisfies $\tau_{ij}=4\det(\rho_i)$,
and equals concurrence-squared, hence the last claim.
\end{proof}

\section{Monotonicity under channels on the target qubit}
\label{sec:target-monotone}

\begin{proposition}[Monotonicity under local noise on the target]
\label{prop:monotone-target}
Fix qubits $(i,j)$ and a projective measurement $\{\Pi^{(i)}_0,\Pi^{(i)}_1\}$ on $i$.
Let $\Phi$ be any CPTP map (quantum channel) acting on qubit $j$ and define
\[
\rho'_{ij} := (I_i\otimes \Phi_j)(\rho_{ij}).
\]
Then
\begin{equation}
K_{i\to j}(\rho'_{ij};\{\Pi_b\}) \;\le\; K_{i\to j}(\rho_{ij};\{\Pi_b\}).
\label{eq:K-monotone-target}
\end{equation}
Consequently, the basis-optimized variant also satisfies
\[
K^{\max}_{i\to j}(\rho'_{ij}) \le K^{\max}_{i\to j}(\rho_{ij}).
\]
\end{proposition}

\begin{proof}
The measurement on $i$ is unchanged, so the outcome probabilities $p_b$ are the same for
$\rho_{ij}$ and $\rho'_{ij}$.
Moreover, the conditional states on $j$ transform as
\[
\rho'_{j|b} = \Phi(\rho_{j|b}) \quad (p_b>0).
\]
Fidelity is monotone non-decreasing under channels~\cite{WatrousQIT}:
\[
\Fid(\rho_{j|0},\rho_{j|1}) \le \Fid(\Phi(\rho_{j|0}),\Phi(\rho_{j|1})) = \Fid(\rho'_{j|0},\rho'_{j|1}).
\]
Therefore $1-\Fid$ can only decrease, and since $4p_0p_1$ is unchanged, \eqref{eq:K-monotone-target}
follows immediately from Def.~1. The statement for $K^{\max}$ follows by taking a maximum over
measurement bases after applying the pointwise inequality for each basis.
\end{proof}

The preceding monotonicity statement is \emph{not} an LOCC monotonicity statement: it only asserts
monotonicity under channels on the \emph{target} qubit $j$ with the measurement on $i$ fixed.
Noise or preprocessing on $i$ can change the branching probabilities and the induced ensemble
on $j$, so no monotonicity should be expected in that direction.

\section{Operational interpretation: guessing the measurement outcome}
\label{sec:operational}

Throughout this appendix, assume $p_0p_1>0$ unless stated otherwise.
When one prior is zero, the measurement outcome is deterministic:
$P_{\mathrm{succ}}^\star=1$ and $K=0$ by Definition~\ref{def:pairwiseK_Z}.
Consequently, outside the balanced case, $K$ alone should not be read as a guessing probability.

\subsection{Helstrom discrimination of the induced ensemble}
The measurement on qubit $i$ induces the binary classical--quantum ensemble
\begin{equation}
\mathcal{E}_{i\to j}(\{\Pi_b\}) := \big\{(p_0,\rho_{j|0}),\, (p_1,\rho_{j|1})\big\},
\label{eq:ensemble}
\end{equation}
where $p_b$ are the measurement outcome probabilities and $\rho_{j|b}$ the corresponding
conditional states.
The optimal success probability for guessing $b$
from a measurement on $j$ is given by the Holevo--Helstrom theorem~\cite{BaeKwekReview,WatrousQIT}:
\begin{equation}
P_{\mathrm{succ}}^\star(\mathcal{E}_{i\to j})
=
\frac{1}{2}\Big(1 + \big\|p_0\rho_{j|0} - p_1\rho_{j|1}\big\|_1\Big).
\label{eq:helstrom}
\end{equation}
This makes precise the statement that the measurement outcome on $i$ leaves an observable imprint on $j$
when the two conditional states $\rho_{j|0}$ and $\rho_{j|1}$ are easy to distinguish.

\subsection{Connecting \texorpdfstring{$K$}{K} to trace distance via Fuchs--van de Graaf}
Let
\[
D := \Tdist(\rho_{j|0},\rho_{j|1}) := \frac{1}{2}\|\rho_{j|0}-\rho_{j|1}\|_1
\]
denote the trace distance between the conditional states. The Fuchs--van de Graaf inequalities~\cite{FuchsVanDeGraaf}
imply the two-sided bound
\begin{equation}
D^2 \;\le\; 1 - \Fid(\rho_{j|0},\rho_{j|1}) \;\le\; 2D - D^2,
\label{eq:fvdg-two-sided}
\end{equation}
valid for all density operators.
Multiplying \eqref{eq:fvdg-two-sided} by $4p_0p_1$ and using Def.~1 yields
\begin{equation}
4p_0p_1\,D^2
\;\le\;
K_{i\to j}(\rho;\{\Pi_b\})
\;\le\;
4p_0p_1\,(2D - D^2).
\label{eq:K-sandwich}
\end{equation}

\paragraph{Balanced case.}
When $p_0=p_1=\tfrac{1}{2}$, we have $K_{i\to j} = 1 - \Fid(\rho_{j|0},\rho_{j|1})$ and
\eqref{eq:K-sandwich} simplifies to
\begin{equation}
D^2 \le K_{i\to j} \le 2D - D^2,
\qquad\text{equivalently}\qquad
1 - \sqrt{1-K_{i\to j}} \le D \le \sqrt{K_{i\to j}}.
\label{eq:D-vs-K-balanced}
\end{equation}
For equal priors, Helstrom reduces to
\[
P_{\mathrm{succ}}^\star = \frac{1}{2}(1+D),
\]
so \eqref{eq:D-vs-K-balanced} gives an explicit interval containing the optimal success probability
consistent with a given $K$ value:
\begin{equation}
1 - \frac{1}{2}\sqrt{1-K_{i\to j}}
\;\le\;
P_{\mathrm{succ}}^\star
\;\le\;
\frac{1}{2}\big(1+\sqrt{K_{i\to j}}\big).
\label{eq:Psucc-vs-K}
\end{equation}

\section{Relation to cq-channel reliability}
\label{sec:cq-reliability}

For $p_0p_1>0$, the measurement on qubit $i$ induces a binary classical--quantum (cq) state on qubit $j$:
\begin{equation}
  \rho_{XB} = \sum_{b\in\{0,1\}} p_b \ket{b}\!\bra{b}_X \otimes \rho_{j|b}.
\end{equation}
In the quantum polar coding literature, a standard quantity is the \emph{reliability parameter}~\cite{RenesPolarCoding,WildeGuha}
\begin{equation}
  Z(X|B)_\rho := 2\sqrt{p_0 p_1}\, f(\rho_{j|0},\rho_{j|1}),
\end{equation}
where $f(\rho,\sigma) := \|\sqrt{\rho}\sqrt{\sigma}\|_1$ is the root fidelity, so $\Fid=f^2$ for the squared fidelity convention used here.
Then
\begin{equation}
  K_{i\to j} = 4p_0 p_1 - Z(X|B)^2.
  \label{eq:K-vs-Z}
\end{equation}
Thus $K$ is the complement of $Z^2$ within the branching envelope $4p_0 p_1$.
When $p_0p_1=0$, Definition~\ref{def:pairwiseK_Z} sets $K=0$; the reliability parameter is then unnecessary because only one branch occurs.

\end{document}